\documentclass[11pt]{article}
\usepackage{color}
\definecolor{darkgreen}{rgb}{0,0.5,0}
\usepackage{hyperref}
\hypersetup{
    unicode=false,          % non-Latin characters in Acrobat’s bookmarks
    colorlinks=true,        % false: boxed links; true: colored links
    linkcolor=red,          % color of internal links (change box color with linkbordercolor)
    citecolor=darkgreen,        % color of links to bibliography
    filecolor=magenta,      % color of file links
    urlcolor=cyan           % color of external links
}

\usepackage[letterpaper,margin=1in]{geometry}
\usepackage{amsmath, amssymb, amsthm, amsfonts}
\usepackage{cite}

\usepackage{bbm}

\usepackage{ifthen}
\usepackage{tikz}
\usetikzlibrary{positioning,decorations.pathreplacing}

\usepackage{appendix}
\usepackage{graphicx}

\usepackage{algorithm}
\usepackage[noend]{algpseudocode}
\usepackage{epstopdf}
\usepackage{wrapfig}
\usepackage{thm-restate}

\usepackage{wasysym}
\usepackage[textsize=tiny]{todonotes}

\newcommand{\alg}{\mathcal{A}}
\newcommand{\stage}{\mathcal{S}}
\newcommand{\reg}{\mathcal{R}}
\newcommand{\dereg}{\mathcal{D}}
\newcommand{\go}{\mathcal{G}}
\newcommand{\polylog}{\mathrm{polylog}}

\usepackage{framed}
\usepackage[framemethod=tikz]{mdframed}
\usepackage[bottom]{footmisc}
\usepackage{enumitem}
\setitemize{noitemsep,topsep=3pt,parsep=3pt,partopsep=3pt}
\usepackage[font=small]{caption}
\usepackage{xspace}

\newtheorem{theorem}{Theorem}[section]
\newtheorem{lemma}[theorem]{Lemma}
\newtheorem{meta-theorem}[theorem]{Meta-Theorem}

\newtheorem{remark}[theorem]{Remark}
\newtheorem{corollary}[theorem]{Corollary}

\newtheorem{definition}[theorem]{Definition}

\usepackage[capitalize, nameinlink]{cleveref}
\Crefname{lemma}{Lemma}{Lemmas}
\Crefname{claim}{Claim}{Claims}
\Crefname{remark}{Remark}{Remarks}
\Crefname{observation}{Observation}{Observations}
\algnewcommand\algorithmicswitch{\textbf{switch}}
\algnewcommand\algorithmiccase{\textbf{case}}

\algdef{SE}[SWITCH]{Switch}{EndSwitch}[1]{\algorithmicswitch\ #1\ \algorithmicdo}{\algorithmicend\ \algorithmicswitch}%
\algdef{SE}[CASE]{Case}{EndCase}[1]{\algorithmiccase\ #1}{\algorithmicend\ \algorithmiccase}%
\algtext*{EndSwitch}%
\algtext*{EndCase}%
\newcommand{\eps}{\varepsilon}

\newcommand{\timeOverhead}{\mathsf{time}\textit{-}\mathsf{overhead}\xspace} 
\newcommand{\msgOverhead}{\mathsf{message}\textit{-}\mathsf{overhead}\xspace} 
\newcommand{\congest}{$\mathsf{CONGEST}$\xspace}
\newcommand{\local}{$\mathsf{LOCAL}$\xspace}

\newcommand{\poly}{\operatorname{\text{{\rm poly}}}}

\renewcommand{\paragraph}[1]{\medskip\noindent {\bf #1}}

\let\oldtextbf=\textbf
\renewcommand\textbf[1]{{\boldmath\oldtextbf{#1}}}

\newcommand{\FullOrShort}{full}

\ifthenelse{\equal{\FullOrShort}{full}}{
	
  \newcommand{\fullOnly}[1]{#1}
  \newcommand{\shortOnly}[1]{}
}{
    \newcommand{\shortOnly}[1]{#1}
    \newcommand{\fullOnly}[1]{}
  }

\begin{document}
\date{}
\title{A Near-Optimal Deterministic Distributed Synchronizer}

\author{
  Mohsen Ghaffari\\
  \small MIT \\
  \small ghaffari@mit.edu
  \and 
  Anton Trygub\\
  \small MIT \\
  \small trygub@mit.edu
 }
%-------------------------------------------------------------------------------
\maketitle

\begin{abstract} 
We provide the first deterministic distributed synchronizer with near-optimal time complexity and message complexity overheads.
Concretely, given any distributed algorithm $\mathcal{A}$ that has time complexity $T$ and message complexity $M$ in the synchronous message-passing model (subject to some care in defining the model), the synchronizer provides a distributed algorithm $\mathcal{A}'$ that runs in the asynchronous message-passing model with time complexity $T \cdot \poly(\log n)$ and message complexity $(M+m)\cdot \poly(\log n)$. Here, $n$ and $m$ denote the number of nodes and edges in the network, respectively. The synchronizer is deterministic in the sense that if algorithm $\mathcal{A}$ is deterministic, then so is algorithm $\mathcal{A}'$. Previously, only a randomized synchronizer with near-optimal overheads was known by seminal results of Awerbuch, Patt-Shamir, Peleg, and Saks [STOC 1992] and Awerbuch and Peleg [FOCS 1990]. We also point out and fix some inaccuracies in these prior works.

As concrete applications of our synchronizer, we resolve some longstanding and fundamental open problems in distributed algorithms: we get the first asynchronous deterministic distributed algorithms with near-optimal time and message complexities for leader election, breadth-first search tree, and minimum spanning tree computations: these all have message complexity $\tilde{O}(m)$ message complexity. The former two have $\tilde{O}(D)$ time complexity, where $D$ denotes the network diameter, and the latter has $\tilde{O}(D+\sqrt{n})$ time complexity. All these bounds are optimal up to logarithmic factors. Previously all such near-optimal algorithms were either restricted to the synchronous setting or required randomization.
\end{abstract}

\newpage

% %-------------------------------------------------------------------------------
% \thispagestyle{empty}

% {   %list of todos page
%     \newpage
%    \listoftodos
%     \thispagestyle{empty}
% }
\newpage{       
%     %table of contents page
%     % \vspace{2cm}    
%     % \bigskip\bigskip\bigskip
\hypersetup{linkcolor=blue}
\small
\tableofcontents
}
\newpage
% %-------------------------------------------------------------------------------
% \setcounter{page}{1}

\section{Introduction and Related Work}
Distributed graph algorithms have been studied extensively over the past five decades. The prevalent models for designing such algorithms are \textit{synchronous} message-passing models, which assume that the processors compute and communicate in lockstep \textit{rounds}. In particular, messages sent over different network links are assumed to traverse at the same speed: per round, all messages are sent at the beginning of the round, and they all arrive by the end of the round. Of course, the reality is far from such a nice and clean picture. In most practical settings, different messages will experience different delays, and even delays that are not predictable to the algorithm.\footnote{In many computer networks---e.g., the Internet---each distributed algorithm is run alongside a myriad of other protocols that are using the network, and the delay in each link is influenced by the amount of congestion on it.} Arguably, the synchronous abstraction is still extremely valuable as it provides a clean context for algorithm design, allowing us to focus on (other) core issues such as graph-theoretic challenges and communication limitations. It also vastly simplifies the analysis because one does not have to deal with all possibilities of the delays experienced by different messages. But can one take these ``lab-grown" algorithms, developed in the sterile synchronous setup, and use them in the asynchronous world?

\textit{Synchronizers} are the fundamental theoretical concept introduced to provide a principled and general answer to this question. Roughly speaking, a synchronizer is a mechanism that can be added to \textit{any} synchronous message-passing algorithm, so that the resulting combination would work correctly in the asynchronous setting. Furthermore, ideally, the synchronizer only incurs a slight loss in the algorithm's efficiency, most notably, its time and message complexities. In this paper, we present the first
deterministic synchronizer with near-optimal time and message complexity overheads. Informally, our synchronizer increases time and message complexities by factors only logarithmic in the network size. This leads to the first near-optimal deterministic asynchronous distributed algorithms for some of the most basic graph problems. 
\vspace{-0.2cm}

\subsection{Models and the synchronizer definition}
\label{subsec:model}
\noindent\textbf{The synchronous message-passing model.} The network is abstracted as an undirected graph $G=(V, E)$, and we use the notations $n:=|V|$ and $m:=|E|$. Each node represents one computer/processor, equipped with a unique identifier, typically assumed to have $b=O(\log n)$ bits. Computation and communications occur in synchronous rounds $1, 2, 3, \dots$. Per round, each computer/node performs some computations on it data. Then, it can send one $B$-bit message to each of its neighbors. Typically, we assume $B=O(\log n)$, and this model variant is usually called \congest~\cite{Peleg2000DistributedApproach}. The variant which allows unbounded message sizes, i.e., $B=\infty$, is often called \local~\cite{Linial1992LocalityAlgorithms, Peleg2000DistributedApproach}. All the messages sent in a round are assumed to arrive by the end of the round. Then, the algorithm proceeds to the next round. When talking about randomized algorithms, we assume that each source has a source of random bits, and these are independent across different nodes.  

When discussing algorithms for graph problems, the standard assumption is that at the start of the algorithm, nodes do not know the global network topology $G$. They only know their own identifier, and perhaps some estimates on basic global parameters, e.g., a constant-factor upper bound on $O(\log n)$. At the end of the algorithm, each node should know its own output. For instance, if we are computing a breadth first search (BFS) tree from a given source node $s$, each node should know its own distance to $s$ and perhaps also its parent in the BFS tree. 

The two main measures of interest for an algorithm are its time complexity and message complexity. For an algorithm $A$ in the synchronous setting, its \textit{time complexity} $T(A)$ is defined as the number of rounds until all nodes generate their output. See the time complexity paragraph of \Cref{subsec:modelSubtleties} for more on this. Its \textit{message complexity} $M(A)$ is defined as the total number of messages sent in the entire network during the algorithm. We comment that sometimes the phrase \textit{communication complexity} is used instead of message complexity.

\paragraph{The asynchronous message-passing model.} In the asynchronous model, there is no notion of rounds (known to the algorithms). Any node can perform some computation on the data it holds and then send messages to its neighbors. These messages are guaranteed to arrive eventually, but there is no time bound known to the algorithm. As such, algorithms are generally described in an \textit{event-driven} language, e.g., ``\textit{upon receiving message ..., do this: ...}.'' 

Analyzing algorithms in the asynchronous model is non-trivial because the model involves non-determinism in the timing of the arrival of different messages. A standard approach~\cite{Peleg2000DistributedApproach}, which makes algorithmic results stronger, is to take a worst-case view: we allow the message delays to be determined by an adversary, who knows the algorithm. To provide a performance bound, we assume an upper bound $\tau$ on the delay of each message (the time from the send event to the corresponding arrive event). The time complexity $T(A)$ of the algorithm $A$ is defined as follows: Let $T$ be the longest possible execution time until all nodes generate their output. This longest possible time is with regard to the worst-case execution of the algorithm, where the adversary controls the delays of different messages, subject to this bound $\tau.$ Then, we define the time complexity $T(A)=T/\tau$, i.e., we normalize the time bound by considering $\tau$ as one time unit. We emphasize that, crucially, the value of $\tau$ is not known to the algorithm. The message complexity $M(A)$ of the algorithm $A$ is defined as the maximum number of messages the algorithm sends in the network, again in the worst-case possible execution where the adversary controls the message delays.

\paragraph{Synchronizers.} Defined initially in the pioneering work of Awerbuch~\cite{Awerbuch1985ComplexitySynchronization}, a synchronizer $\mathcal{S}$ is an algorithmic module that can be added to \textit{any} synchronous distributed algorithm $\mathcal{A}$ for any problem $\mathcal{P}$ so that the combined algorithm $\mathcal{A}'$ would be guaranteed to run correctly in the asynchronous environment and solve $\mathcal{P}$. The key measures of interest for a synchronizer are its \textit{time complexity overhead} and \textit{message complexity overhead}, as well as its \textit{initialization time complexity} and \textit{initialization message complexity}. We define these next. Let us first ignore the initialization part and assume that the synchronizer initialization has already been performed. Then, the \textit{time complexity overhead} of $\mathcal{S}$ is defined as $\timeOverhead(\mathcal{S})={T(\mathcal{A}')}/{T(\mathcal{A})}$. Similarly, the \textit{message complexity overhead} of $\mathcal{S}$ is defined as $\msgOverhead(\mathcal{S})={M(\mathcal{A}')}/{M(\mathcal{A})}$. Besides these, the synchronizer may have some initialization phase, which is an algorithm that should run in the asynchronous setting and set up some structures before the start of the synchronizer-augmented variant of $\mathcal{A}$. This initialization itself has time complexity $T^{init}(\mathcal{S})$ and message complexity $M^{init}(\mathcal{S})$. Then, the synchronizer should satisfy the following for every synchronous algorithm $\mathcal{A}$: $$T(\mathcal{A}') \leq T^{init}(\mathcal{S}) + \timeOverhead(\mathcal{S}) \cdot T(\mathcal{A})\textit{, and}$$
$$M(\mathcal{A}') \leq M^{init}(\mathcal{S}) + \msgOverhead(\mathcal{S}) \cdot M(\mathcal{A}).$$

\subsection{State of the art}
\noindent\textbf{Synchronizers with global pulse generation per round.} A natural approach to building synchronizers, set forth by Awerbuch~\cite{Awerbuch1985ComplexitySynchronization}, is to have each node $v$ \textit{generate a pulse} for each round. Concretely, node $v$ will generate pulses $1$, $2$, $3$, \dots, and these pulses demarcate for node $v$ the transition between consecutive rounds. These pulses are known only to $v$. A message is called a pulse $p$ message if it is sent by node $v$ between its pulses $p$ and $p+1$. To simulate the synchronous algorithm, each node $v$ should generate its pulse $p+1$ only after its pulse $p$ and, more crucially, after $v$ has received every message of pulse $p$ sent to it by each neighbor $u$. Of course, the challenge is that node $v$ does not know which neighbors send messages to it, and simply waiting for a predetermined amount of time cannot help; some messages might take unpredictably long. 

Awerbuch introduced three synchronizes, known as $\alpha$, $\beta$, and $\gamma$, which follow this outline of each node generating all the pulses $1$, $2$, $3$, \dots. We review these in \Cref{app:alpha-beta-gamma}. Here, we only note that even ignoring their initialization costs, these synchronizers do not achieve a $\poly(\log n)$ time and message complexity overhead. Indeed, in each case, at least one of these two overheads is $\Omega(n)$. Moreover, this problem is inherent to the approach that tries to generate all pulses at all nodes. This overly strong requirement necessitates some communication for each round/pulse for each node, while the synchronous algorithm might have each node send only in very few rounds.

\paragraph{Synchronizers with succinct pulse generation per round.}
Intuitively, to have small complexity overheads, each node should demarcate only rounds in which the node has some message to send/receive in the synchronous algorithm. In a pair of marvelous papers~\cite{Awerbuch1990SparsePartitions, Awerbuch1990NetworkOverhead}, in 1990, Awerbuch and Peleg stated the existence of synchronizers with only $\poly(\log n)$ time and message complexity overheads. The first paper~\cite{Awerbuch1990SparsePartitions} presented the graph-theoretic notions of sparse partitions/covers and gives sequential algorithms for their construction. The other~\cite{Awerbuch1990NetworkOverhead}, published in the same venue, sketched that once such sparse covers are constructed, one can transform any synchronous algorithm into an asynchronous algorithm at the cost of only $\poly(\log n)$ time and message complexity overheads. However, there are two drawbacks. 

The first drawback is that the sparse cover construction provided in~\cite{Awerbuch1990SparsePartitions} was sequential and would require a high time and message complexity. A follow-up work of Awerbuch, Patt-Shamir, Peleg, and Saks~\cite{Awerbuch1992AdaptingNetworks} provided an efficient \textit{randomized} asynchronous distributed construction, based on the synchronous low-diameter decomposition of Linial and Saks~\cite{Linial1993LowDecompositions}. This led to the first randomized distributed synchronizer with only $\poly(\log n)$ time and message complexity overheads. However, it remained open whether there is an efficient deterministic synchronizer. 

The second drawback is, unfortunately, inaccuracy or incorrectness issues. The papers ~\cite{Awerbuch1990NetworkOverhead, Awerbuch1992AdaptingNetworks} involve several ingenious ideas and a generally viable approach to building synchronizers with small overheads. Howoever, regrettably, the only available versions of ~\cite{Awerbuch1990NetworkOverhead, Awerbuch1992AdaptingNetworks} are their conference versions.\footnote{We have had personal communication with David Peleg, who confirmed that no full version of these papers exists.} In many places, these versions do not provide proofs, or even concrete lemma statements stating the desired and guaranteed behavior of parts of the algorithm. This leaves much to be filled out by the reader, and to the best of our understanding, there are at least two critical inaccuracies:

(I) The synchronizer scheme presented in \cite{Awerbuch1990NetworkOverhead}, as described, is technically incorrect. It overlooks a small but critical issue of congestion, which, once corrected, increases the time complexity overhead to $\Omega(n)$. The follow-up work of Awerbuch et al.~\cite{Awerbuch1992AdaptingNetworks}, mentioned above, described the synchronization approach of~\cite{Awerbuch1990NetworkOverhead} in a different way, which incidentally goes around the congestion issue. This was without explicitly pointing out any incorrectness issues in the latter. But this adaptation opens up a subtle correctness problem. One needs a highly non-trivial statement and analysis to show that the algorithm still operates as desired despite the asynchrony and all possibilities of message delays. In fact, we do not see how to fix the algorithm in~\cite{Awerbuch1992AdaptingNetworks} in the exact way that it is written.

(II) The synchronizers in \cite{Awerbuch1990NetworkOverhead, Awerbuch1992AdaptingNetworks} are described primarily for synchronizing a BFS algorithm from a given source. The above incorrectness/incompleteness issues apply even to this restricted case. But the work of \cite{Awerbuch1990NetworkOverhead, Awerbuch1992AdaptingNetworks} claims synchronization for an arbitrary synchronized algorithm. Both papers contain a brief follow-up section, after discussing BFS, that claims how to generalize the BFS approach (without much proof). We argue that this extension's correctness depends on the interpretation of algorithms in the synchronous model. If in the synchronous algorithm, a node is allowed to do something like ``\textit{wait for $r$ many rounds, or wait for round number $r$, and then send  message $m$}'', we argue that the message complexity overhead of the synchronizers of \cite{Awerbuch1990NetworkOverhead, Awerbuch1992AdaptingNetworks} can be $\Omega(n)$. See \Cref{subsec:modelSubtleties} for more on this.

\subsection{Our contribution}
To summarize, there are two drawbacks to the state-of-the-art: (a) the synchronizers of ~\cite{Awerbuch1992AdaptingNetworks, Awerbuch1990NetworkOverhead} are randomized, (b) they contain incompleteness/incorrectness issues. We address both.

\subsubsection{A near-optimal deterministic distributed synchronizer.} Developing \textit{deterministic} distributed algorithms for graph problems has been one of the leading research themes over the past four decades, with significant recent breakthroughs; see, e.g.,~\cite{Rozhon2020Polylogarithmic-timeDerandomization, Ghaffari2018OnAlgorithms, Ghaffari2017DistributedTime}. However, without a deterministic synchronizer, these deterministic synchronous algorithms cannot be transported to the more realistic asynchronous world (without significant overheads). Obtaining a deterministic synchronizer with $\poly(\log n)$ time and message complexities has been a long-standing open problem since~\cite{Awerbuch1992AdaptingNetworks}. See e.g.~\cite{ElkinSynchronizers}. 
Our main contribution is to remedy this. We provide the first deterministic distributed synchronizer with $\poly(\log n)$ time and message complexities:

\begin{theorem} \label{thm:mainInformal}(\textbf{Informal}) There is a deterministic synchronizer, with no initialization, with time and message complexity overheads $\poly(\log n)$. Concretely, any synchronous algorithm can be transformed into an asynchronous algorithm with a $\poly(\log n)$ loss in time and message complexities.\footnote{Technically, we assume that the synchronous algorithm sends at least one message along each edge, and thus has message complexity $\Omega(m)$. Without this, we should state the message complexity of the asynchronous version with an additive $m\poly(\log n)$ term, as done in the abstract. The same condition applies to the synchronizers of ~\cite{Awerbuch1992AdaptingNetworks, Awerbuch1990NetworkOverhead}, and this $\Omega(m)$ message complexity is satisfied by a wide range of graph algorithms of interest.} 
\end{theorem}

\paragraph{Applications.} Our near-optimal deterministic synchronizer leads to the first deterministic distributed asynchronous algorithms with near-optimal time and message complexities for many of the basic graph problems, including single-source or multi-source BFS, leader election, and minimum spanning tree, as we state in corollaries next. See \Cref{app:applications} for proofs of these corollaries.

\begin{corollary}
    \label{crl:BFS}
    There is a deterministic asynchronous distributed algorithm that, given a source node $s$, computes a breadth first search tree rooted in this source node in $\tilde{O}(D)$ time and using  $\tilde{O}(m)$ messages. Each node $v$ learns its distance from source $s$ and knows its parent in the BFS tree. The algorithm can be extended to the multi-source setting, where there is a set $S$ of multiple sources, and each node should join the BFS tree of the closest source.
\end{corollary}

\begin{corollary}
    \label{crl:LE}
    There is a deterministic asynchronous distributed algorithm that elects a leader in the network in $\tilde{O}(D)$ time and using  $\tilde{O}(m)$ messages. Every node learns the identifier of the leader.
\end{corollary}

We remark that for leader election, even the randomized asynchronous case was open until a very recent work of Kutten et al.\cite{Kutten2021SingularlyNetworks}, and their randomized algorithm assumes that nodes know $n$ up to a constant factor. Our algorithm is deterministic, and the more standard assumption of nodes known $n$ up to a polynomial suffices for it, i.e., knowing a constant-factor upper bound on $O(\log n)$, which is the identifier/message sizes and needed as part of the \congest model. 

\begin{corollary}
    \label{crl:MST}
    There is a deterministic asynchronous distributed algorithm that computes a minimum spanning tree in $\tilde{O}(D+\sqrt{n})$ time and using  $\tilde{O}(m)$ messages.
\end{corollary}

As a side remark, we note that for asynchronous MST, even the randomized case was open. Mashreghi and King~\cite{Mashreghi2019BriefCommunication} asked whether there is a randomized asynchronous MST algorithm with $\tilde{O}(m)$ messages that takes even just $o(n)$ time in graphs that have $D=o(n)$. A very recent work of Dufoulon et al.~\cite{Dufoulon2022AnAlgorithm} gives a randomized algorithm with $\tilde{O}(m)$ messages that take even just $\tilde{O}(D^{1+\eps}+\sqrt{n})$ time, for any small constant $\eps>0$. \Cref{crl:MST} is deterministic and achieves the optimal time and message complexity bounds up to logarithmic factors.

\subsubsection{Inaccuracy or correctness issues.} As mentioned before, there are two types of issues in \cite{Awerbuch1992AdaptingNetworks, Awerbuch1990NetworkOverhead}: (I) problems that appear even for synchronizing BFS algorithms (concretely in a  subroutine for registration and deregistration), and (II) problems in the extension to general synchronous algorithms (concretely, in the definition of the synchronous model and the message complexity blow-up of the extension). 

For issue (I), the fix needs some algorithmic adjustments and careful lemma statements. We describe a variant of the approach of \cite{Awerbuch1992AdaptingNetworks, Awerbuch1990NetworkOverhead} and provide a correctness analysis that considers all possibilities in the asynchronous environment. See \Cref{subsec:registration} for more on this.

Issue (II) turns out to depend critically on the interpretation of synchronous algorithms. We are not aware of any explicit discussion about this subtlety in prior work. We show a natural and useful interpretation with which the $\poly(\log n)$ time and message complexity overheads claims remain correct. Informally, the synchronous algorithm should have no explicit reference to the round numbers. It should be only \textit{event-driven}, meaning it is described as: \textit{Upon receiving messages \dots (potentially several in the same round), send messages \dots (immediately after, as soon as possible for the given number of messages to be sent, and without waiting for many rounds).} See \Cref{subsec:modelSubtleties} for more on this. Fortunately, a wide range of synchronous algorithms of interest are event-driven in this sense or can be paraphrased to be event-driven without any asymptotic loss in their complexities. This includes usual algorithms for BFS, leader election, and MST.\footnote{Furthermore, even algorithms that depend on time can be turned into this event-driven language, by each node generating a clock for itself by communicating back and forth with one of its neighbors, but that would be at the expense of $nT$ additional message complexity, where $T$ denotes the time complexity of the synchronous algorithm.} 

Finally, we note another subtle point that is not explicitly discussed in~\cite{Awerbuch1992AdaptingNetworks, Awerbuch1990NetworkOverhead}: The time complexity definition for which the $\poly\log n$ time complexity overhead holds is the time from the start until all nodes generate their output (in both the synchronous and asynchronous algorithms). In the asynchronous version of a synchronous algorithm generated by the synchronizers of ~\cite{Awerbuch1992AdaptingNetworks, Awerbuch1990NetworkOverhead}, it is possible that some nodes might continue to perform some auxiliary communications longer, after outputting their output from the synchronous algorithm (or when in the synchronous algorithm they have no output). This can continue for up to $\tilde{O}(D)$ time. Our deterministic synchronizer is also with respect to the same definition of time complexity. See the time complexity paragraph in \Cref{subsec:modelSubtleties} for more on this.

\section{Preliminaries}
\noindent\textbf{Notations.} We assume that the network is an $n$-node $m$-edge undirected graph $G=(V, E)$, and we use $D$ to denote the diameter of the graph. We use $dist(v, u)$ to denote the distance of nodes $v$ and $u$ in the graph. For a set $S$ of vertices, we define $dist(v,S):=min_{s\in S} dist(v, s)$.

\subsection{Sparse covers}
\label{subsec:sparseCovers}
 
\begin{definition} For any positive integer $d$, we define a \textbf{sparse $d$-cover} with stretch $s$ of an $n$-node graph  to be a set of clusters, such that the following conditions hold:

\begin{itemize}
    \item The diameter of each cluster is $O(d\cdot s)$.

    \item Each node is in $O(\log{n})$ clusters.

    \item For any $u, v \in V$ s.t. $dist(u, v)\le d$, there is at least one cluster that includes both $u$ and $v$. In fact, a stronger statement holds. For each node $v$, there exists a cluster such that all node $u$ for which $dist(u, v)\leq d$ are included in this cluster.
\end{itemize}
We assume that each node knows in which clusters it is. Furthermore, for each cluster, we have a tree of depth $O(d\cdot s)$ that spans all cluster nodes, where each node knows its parent and its children.

Often we need sparse covers for all powers of $2$ up to roughly $d$. To capture that, we define a \textbf{layered sparse $d$-cover} for a parameter $d$ to be a collection of sparse $2^{j}$-covers for all $j\in \{0,, 1,2, \dots, \lceil{\log_2 d\rceil}\}$.
\end{definition}

The optimal value of stretch in sparse covers is $O(\log n)$~\cite{Awerbuch1990SparsePartitions, Linial1993LowDecompositions}.
We work with deterministic distributed constructions of sparse cover that, while being time efficient, have slightly higher stretch~\cite{Rozhon2020Polylogarithmic-timeDerandomization}. For the rest of the paper, we will build and use only sparse $d$-covers with the properties stated below. See \Cref{theorem:synchcover} and \Cref{theorem:sparseasynch} for the precise theorem statements in the synchronous and asynchronous settings.

\begin{itemize}
    \item The cluster stretch is $O(\log^3{n})$. More concretely, each cluster has a $O(d\log^3{n})$-depth cluster tree, which spans all cluster nodes.
    \item Each edge is used in only $O(\log^4{n})$ cluster trees.
\end{itemize}

\subsection{Running several algorithms in the asynchronous environment}

In our algorithms, we often have several algorithmic subroutines. A scheduling question arises when these subroutines want to use the same edge for communications. In general, we might need to run all these subroutines in parallel in the sense that each of them relies on the progress of the others, or we might want to run them sequentially, in the sense that the $i^{th}$ one depends only on subroutines $1$ to $i-1$. In the synchronous setting, this is usually simple. In this subsection, we discuss how to perform these in the asynchronous environment, and the related time guarantees.

\subsubsection{Running in parallel}
\label{subsubsec:parallel}
 Suppose we want to run several subroutines in parallel. These subroutines may desire to communicate over the same edge $e$. Notice that, in both the synchronous setting (with the event-driven interpretation) and the asynchronous setting, we have restricted the algorithms to waiting for each injected message's acknowledgment, before sending the next message. Hence, different subroutines might appear to slow down each other. How do we schedule the communications of these subroutines over edge $e$, so they all work in a small time complexity? 

Let us call the subroutines $\alg_1, \alg_2, \ldots, \alg_t$. It would be very convenient if we had a separate copy of the edge $e$ for each subroutine. Then the messages sent for different subroutines would not get in the way of each other. 
To deal with these issues, we will use the following abstraction principle to simulate ``making copies" of the edges:

\begin{lemma}
    \label{lemma:duplicated}  For a given graph $G$ and an integer $k$, let us define $G(k)$ as a graph $G$, in which each edge was duplicated $k$ times. In particular, $G(1) = G$. Suppose that for a given $k$, some asynchronous algorithm $\alg$ runs on the graph $G(k)$ in time $O(f)$, where $f$ is some function from the parameters of the graph. Then, we can simulate algorithm $\alg$ on the graph $G$ in time $O(kf)$.    
\end{lemma}

\begin{proof}
    For each node $u$, and any node $v$ connected to it, let $u$ number edges from $u$ to $v$ in $G(k)$ by $1, 2, \ldots, k$. $u$ will simulate sending messages to $v$ in $G(k)$ by sending messages in turns. At the $i$-th turn, $v$ will check if there is a message that $v$ wants to send to $u$ through the $(i \bmod k)$-th edge. If yes, $v$ will send this message and proceed to the $i+1$-st turn; otherwise, it will proceed to the $i+1$-st turn right away.

    All the nodes are sending the same messages as they would according to $\alg$ in $G(k)$. The only difference is that while previously, they were guaranteed to be delivered in $1$ time unit, now they are guaranteed to be delivered in $k$ time units, as a message has to wait for at most $k$ turns. So, the time complexity of the algorithm increases by at most a factor of $k$.
\end{proof}

\begin{corollary}
\label{corollary:parallel}
Suppose that we want to run several procedures $\alg_1, \alg_2, \alg_3, \ldots$ on a graph $G$, and the following conditions hold:
    (a) Each edge is involved in sending messages for $O(k)$ procedures. An edge might not know in advance in which procedures it is involved; see \Cref{rmrk:crlParallelAdaptivity} for more. (b) If for each edge $e$ we had a separate copy of $e$ for every procedure that involves $e$, then all the procedures would be completed in time $O(t)$.
    Then, we can run all these procedures on $G$ in $O(kt)$ time.
\end{corollary}

\begin{remark}\label{rmrk:crlParallelAdaptivity}
    Note that the edge $(u, v)$ does not have to know in advance how many procedures it will be involved in. The statement holds if each edge is going to be involved in sending messages of at most $O(k)$ procedures, even if it does not know whether it will receive messages of new procedures in the future and does not know the value of $k$. Each endpoint of $e$ will simply perform communications of all of the procedures that want to send a message along $e$ in a round-robin fashion, along this edge $e$. In this sense, this abstraction is adaptive.
\end{remark}

\subsubsection{Running sequentially}

\label{subsubsec:sequential}

Suppose now that our algorithm consists of several stages $\stage_1, \stage_2, \ldots, \stage_k$, where $\stage_i$ relies only on the results of previous stages. In the synchronous world, if we knew the time bounds on each stage, we could run them sequentially: before starting the stage $\stage_{i+1}$, we would wait until the entire $\stage_i$ is finished. Unfortunately, in the asynchronous world, we cannot have any such notion of "waiting" for a certain amount of time, for the previous stages to be over. Note that we do not require the entire $\stage_i$ to be finished before we start sending any of the messages of $\stage_{i+1}$: messages may go through faster in some parts of the graph than in others. 

How do we bound the time spent on running all of these stages? With the approach from \Cref{subsubsec:parallel}, we could do it as follows: Suppose that if for each edge we had a separate copy of it for each stage, all the procedures would be over in time $O(t)$. Then the best we could say according to the \Cref{corollary:parallel} is that all stages would be over in $O(kt)$ time. We, however, can prove a stronger statement.

\begin{lemma}
    \label{lemma:sequential}
    Suppose that there is an algorithm $\alg$ that involves several sequential stages $\stage_1, \stage_2, \ldots, \stage_k$, and the algorithm works on a graph $G$, with the following properties:
    \begin{itemize}
        \item Each $\stage_i$ relies only on the previous stages $\stage_1, \dots, \stage_{i-1}$. In other words, it is possible that $\stage_i$ starts after $\stage_1, \dots, \stage_{i-1}$ are completed, and no stage $\stage_{j}$, for $j>i$, has even started yet.

        \item If all stages $\stage_1, \stage_2, \dots, \stage_{i-1}$ are over, and $\stage_i$ has not started yet, then $\stage_i$ is over in at most $T_i$ additional time. We will refer to $T_i$ as its \textbf{isolated time complexity}.
    \end{itemize}
    Then, we can run this algorithm in time $O(T_1 + T_2 + \ldots + T_k)$.
\end{lemma}

\begin{proof}
We run the algorithms together, with the following scheduling rule for messages: whenever a node has several messages to send over an edge $e=\{v, u\}$, node $v$ prioritizes messages of lower stages and sends them first. With this rule, the only way for a message of $\stage_{i}$ to be delayed is that it is waiting for some messages of stages $\stage_1, \stage_2, \dots, \stage_{i-1}$ to go through the edge. Notice that since the stages are devised for the asynchronous setting, each of them should work correctly even if its messages are delayed adversarially by an arbitrary amount of time. Now, by an induction on $i$, we argue that by time $T_1 + T_2 + \ldots+T_i$, all stages $\stage_1, \stage_2, \ldots, \stage_i$ are over. For the base case, the algorithm of $\stage_1$ is not delayed by any other stage and thus finishes in $T_1$ time. For the inductive step, notice that by time $T_1 + T_2 + \ldots+T_j$, all stages $\stage_1, \stage_2, \ldots, \stage_j$ are over. By that point of time, the algorithm of $\stage_{j+1}$ might have made some partial progress in its execution, but it always remains in a valid state of the algorithm. Hence, by the second condition in the lemma statement, we know that from any valid state, the $\stage_{j+1}$ algorithm takes at most $T_{j+1}$ more time to finish. The reason is this: consider the adversarial execution where we first execute and complete  $\stage_1, \stage_2, \ldots, \stage_j$, then we start $\stage_{j}$ but we adjust the delays to reach the same valid state, and then we run the rest of the algorithm; this should still terminate within $T_{j+1}$ more time. Hence, all stages $\stage_1, \stage_2, \ldots, \stage_{j+1}$ are over by time $T_1 + T_2 + \ldots+T_{j+1}$. This completes the inductive proof.
\end{proof}

Our entire algorithm will consist of several stages, satisfying the constraints of the \cref{lemma:sequential}, so we will bound its total runtime by the sum of the isolated time complexities of all its stages.

\section{Collecting information in covers}
\subsection{Collecting information in $(d \cdot \ell)$-ball with sparse $d$-cover}
\label{SUBSEC:BALLCOLLECTING}

In the course of our algorithm, sometimes we need a node to gather some information about other nearby nodes. We use the notion of sparse covers described in \Cref{subsec:sparseCovers} to do this efficiently. 

Consider some algorithmic subroutine $\mathcal{P}$. This can be simply a part of the algorithm, e.g., a phase in the synchronous algorithm. Suppose that each node will run this process $\mathcal{P}$ at most once. We will say that node $v$ \textbf{is done with $\mathcal{P}$} when process $\mathcal{P}$ terminates in node $v$, or when node $v$ learns that $v$ will not run $\mathcal{P}$ at all. In this subsection, we discuss a way to collect information about the completion of $\mathcal{P}$ using sparse $d$-covers, assuming these covers have already been constructed.

\paragraph{Information gathering up to the sparse cover radius.} Let us say that we want each node $v$ to learn when all the nodes in $d$-neighborhood of $v$ are done with $\mathcal{P}$. The following theorem statement provides the desired procedure.

%in a setting where there is a guarantee of the form "by a time $t$, every node will be done with $\mathcal{P}$". Here nodes do not have to know this value $t$. We will collect this information with sparse covers, which we assume to be given.

\begin{theorem}
    \label{theorem:convergecast}
    Under the conditions above, there exists a procedure $\alg$ that lets each node learn when all the nodes in its $d$-neighborhood are done with $\mathcal{P}$. The procedure has the following properties:

    \begin{itemize}
        \item Let $t$ be the time by which all nodes are done with $\mathcal{P}$ (the value of $t$ is not known only for the analysis and is unknown to the algorithms). Then $\alg$ terminates by time $t + O(d\log^7{n})$.

        In other words, the isolated time complexity of collecting this information is $O(d\log^7{n})$.
        \item Procedure $\alg$ uses only $O(m\log^4{n})$ extra messages. 
    \end{itemize}
    
    %Consider some process $\mathcal{P}$, such that each node will do it at most once. We suppose that we are given a sparse $d$-cover. Then there exists an algorithm $\alg$, with which we can and that each node is guaranteed to be done with $\mathcal{P}$ by the time $t$. Then, in time at most $t + O(d\log^3{n})$, we can let each node know that all the nodes in its $d$-neighborhood are done with $\mathcal{P}$, with only $O(m\log^4{n})$ extra messages.
\end{theorem}

\begin{proof}
    We do a convergecast in each cluster of the cover separately. Consider one particular cluster. We use a convergecast and broadcast in the cluster. In particular, node $v$ waits until all the nodes in its subtree are done with $\mathcal{P}$. If $v$ has any children, it gets this update from its children. Once all descendants of $v$ are done with $\mathcal{P}$, node $v$  notifies its parent that $v$ and its subtree are done with $\mathcal{P}$. Once the cluster root has learned that the entire cluster is done with $\mathcal{P}$, it starts a broadcast, propagating the confirmation that everyone in the cluster is done with $\mathcal{P}$. Each node that receives this confirmation from its cluster parent sends it to its children in the cluster. If the node is a leaf, then it does nothing (effectively, the broadcast terminates in its branch).  

    Now let us consider a node $v$. Notice that the node is included in many clusters. Node $v$ waits until it receives such a confirmation broadcast from all the clusters in which $v$ is present. Since for every two nodes $u, v$ at a distance at most $d$ there is a cluster in which both of them are present, node $v$ will finish collecting the information only after $u$ is done with $\mathcal{P}$. 
    
    We can imagine this as an algorithm consisting of two parts: the first part is running $\mathcal{P}$, and at the end of it each node should know that itself is done with $\mathcal{P}$. The second part is ensuring that we wait until all nodes up to distance $d$ are done with $\mathcal{P}$. This second part is done by a convergecast and broadcast in the cluster trees. Remember that each convergecast/broadcast is up to distance $O(d\log^3{n})$, and each edge is in $O(\log^4{n})$ cluster trees, so, by \cref{corollary:parallel}, the isolated time complexity of the second part is $O(d\log^7{n})$. Since the first part is not dependent on the second one, by \cref{lemma:sequential}, we get that, if the first part finishes by time $t$, then the entire algorithm finishes by time $t + O(d \log^3{n} \cdot \log^4{n}) = t + O(d\log^7{n})$.

    For message complexity, note that each edge is in $O(\log^4{n})$ clusters, each cluster sends only $O(1)$ messages through each of its edges. Thus, we have only $O(m\log^4{n})$ extra messages.
\end{proof}

\paragraph{Information gathering slightly above the sparse cover radius.}
What if we want to collect information up to distance of $d\ell$, which is somewhat greater than the radius $d$ that is guaranteed to be covered by the provided sparse cover? We show that this is also possible, though we end up paying an $\ell$ factor in time and message complexities. The time complexity increase is natural as we care about the greater radius and it takes $\Omega(d\ell)$ time for any information to traverse this distance. The message complexity increase is in some sense an inefficiency of this extension. The following theorem statement provides the desired procedure.

\begin{theorem}
    \label{theorem:largeconvergecast}

    Under the conditions above there exists a procedure $\alg$, letting each node learn when all the nodes in its $d\ell$-neighborhood are done with $\mathcal{P}$, with the following properties:

    \begin{itemize}
        \item If all nodes are guaranteed to be done with $\mathcal{P}$ by some time $t$, $\alg$ is guaranteed to be over in time $t + O(d\ell\log^7{n})$.
        \item It takes only $O(m\ell\log^4{n})$ extra messages.
    \end{itemize}
\end{theorem}

\begin{proof}

Let process $\mathcal{P}_\ell$ be: ``knowing that my $d\cdot \ell$-neighborhood is done with $\mathcal{P}$''. Then $v$ can finish process $\mathcal{P}_{\ell+1}$ by collecting the information about completion of $\mathcal{P}_\ell$ in its $d$-cover. So, in addition to waiting until nodes complete the initial process $\mathcal{P}$, we sequentially run $\ell$ stages. In the $i$-th stage, we make all nodes be done with $\mathcal{P}_i$.

We then apply the same argument as in \cref{theorem:convergecast}, and get that the entire algorithm will be over by time $t + O(\ell\cdot d\log^3{n} \cdot \log^4{n}) = t + O(d\ell\log^7{n})$. As we do $\ell$ stages, we have only $O(m\ell\log^4{n})$ extra messages.
%
%directly. Suppose that for some integer $l$ we can let all nodes know that their $d\cdot l$-neighborhood is done with $\mathcal{P}$ in at most $t + O(d\log^3{n} \cdot l)$ time with only $O(m\log^4{n} \cdot l)$ extra messages. Let process $\mathcal{P}_l$ be: ``knowing that my $d\cdot l$-neighborhood is done with $\mathcal{P}$". Then, in just $O(d\log^3{n})$ extra time and $O(m\log^4{n})$ extra messages, we can let each node know that its $d$-neighborhood is done with $\mathcal{P}_l$, by \Cref{theorem:convergecast}. Note that if all nodes in a $d$-neighborhood of a node $v$ are done with $\mathcal{P}_1$, all nodes in a $d + d\cdot l$-neighborhood of $v$ have to be done with $\mathcal{P}$. So, we can let all nodes know that their $d\cdot (l+1)$-neighborhood is done with $\mathcal{P}$ in at most $t + O(d\log^3{n} \cdot (l+1))$ time with only $O(m\log^4{n} \cdot (l+1))$ extra messages. 
\end{proof}

% \newpage

\subsection{Registration in clusters}
\label{subsec:registration}

\subsubsection{The Abstraction}

In \Cref{SUBSEC:BALLCOLLECTING}, we considered a process $\mathcal{P}$ such that each node $v$ eventually learns that $v$ is done with $\mathcal{P}$, i.e., that either $v$ participated in $\mathcal{P}$ and it has terminated, or $v$ learned that $v$ will not run process $\mathcal{P}$. Here, we discuss a somewhat similar setup, but with a critical difference: a node $v$ does not know whether it will participate in $\mathcal{P}$ or not, and it may never participate in process $\mathcal{P}$. In an informal sense, now node $v$ is done with $\mathcal{P}$ only if it participates and has terminated process $\mathcal{P}$. We would again like to gather information about nodes that are done with $\mathcal{P}$ (this will be primarily relevant only for other nodes that are done with $\mathcal{P}$ themselves). The precise properties that we desire are more subtle and are described below.

\begin{definition}[\textbf{Registration Abstraction}]
\label{registrationAbstraction}
    Consider some cluster $C$, whose nodes are performing some process $\mathcal{P}$. We will use the following abstraction to collect the information about the completion of the processes $\mathcal{P}$. When node $v$ wants to start the process $\mathcal{P}$, it will \textbf{register} in the cluster $C$. When node $v$ has finished the process $\mathcal{P}$, it will \textbf{deregister} in the cluster $C$ and wait until $C$ sends it a message \textbf{$Go\_Ahead$}, allowing it to proceed to the next steps of the algorithm. We emphasize that a node does not necessarily know in advance whether it will take part in $\mathcal{P}$ or not; it might end up participating in it because of receiving some message.
\end{definition}

Ideally, we would want node $v$ to receive $Go\_Ahead$ after all the nodes who are going to register in the cluster have already registered and deregistered. That is, we do not want $v$, who has participated in $\mathcal{P}$, to proceed to the next steps of the algorithm if some node from some cluster $C$ that includes $v$ is still performing $\mathcal{P}$ and has not terminated $\mathcal{P}$. We will state the precise desired guarantees soon.

Let us also briefly discuss the desired complexities: Intuitively, we also desire that each registration/deregistration operation takes time only proportional to the height of the cluster tree, and that the message complexity of the solution is proportional to the number of the nodes that register/deregister (up to a factor of the height of the cluster), but critically not proportional to all the nodes/edges in the entire graph. We discuss the precise complexity bounds of our solution in \Cref{subsec:registrationAnalysis}.

\paragraph{Is a simple convergecast/broadcast enough?} We could try to collect the above information with a convergecast in the tree: node gathers the information about registration/deregistration inside its subtree and passes it to its parent. When everyone who is going to register has already deregistered, the root would issue a $Go\_Ahead$, and propagate it to all the nodes who have participated in the registration process. Unfortunately, this is not possible in this simple way: a node does not know in advance whether it will perform $\mathcal{P}$ or not. Hence, it cannot know in advance whether it will register or not, and its parents do not know in advance whether their subtree is done running $\mathcal{P}$ or not, and simply waiting will not solve the problem. This is because it is possible that one node never participates in $\mathcal{P}$ and will never know that it will not participate.

\paragraph{Provided guarantees by the registration process.} We will ask for two guarantees from this registration process:

\begin{restatable}[\textbf{Register Guarantee 1}]{lemma}{registerGuaranteeOne}
\label{lemma:guarantee1}
    When node $v$ receives $Go\_Ahead$, all the nodes that have registered before $v$ has deregistered are already deregistered.   If the distance from the node $v$ to the root $r$ is $h$, then both its registration and deregistration take $O(h)$ time and messages.
\end{restatable}

\begin{restatable}[\textbf{Register Guarantee 2}]{lemma}{registerGuaranteeTwo}
\label{lemma:guarantee2}
     Suppose eventually no new registration happens, and all nodes who have registered have deregistered. Then, if the height of the tree is $h$,  then every node that has ever registered will receive $Go\_Ahead$ in time $O(h)$. The total number of messages spent on sending $Go\_Ahead$s is proportional to the total number of messages spent on registration and deregistration.
\end{restatable}

% \begin{enumerate}
%     \item \textbf{Register Guarantee 1:} 
%      When node $v$ receives $Go\_Ahead$, we have this: consider any node $u$ that has completed its registration before the time that $v$ deregistered. Then, it must be the case that $u$ has deregistered, before $v$ receives $Go\_Ahead$.

%     \item \textbf{Register Guarantee 2:}    
%     If there is a point in time after which no new registration happens, and all nodes who have registered have deregistered, then every node that has ever registered will receive $Go\_Ahead$.
% \end{enumerate}

\paragraph{A natural attempt.} A natural way to implement the operations of registration, deregistration, and sending $Go\_Ahead$ would be to send all the messages directly to the cluster root, as follows:

\begin{itemize}

\item

To register node $v$ in the cluster $C$, send a registration message to the root of $C$ with the id of $v$ and let the root remember it. Once $v$ receives a response back from the root, it is considered registered.

\item

To deregister node $v$, send a deregistration message to the root, again with the id of $v$, and let the root remember it. Once $v$ receives a response back from the root, it is considered deregistered.

\item

When the root detects that all registration messages have been matched with all deregistration messages, it will issue a $Go\_Ahead$ and propagate down the tree to all the nodes that are waiting for it.

\end{itemize}

Indeed, this is essentially the version of registration and deregistration described in the pioneering work of Awerbuch and Peleg~\cite{Awerbuch1990NetworkOverhead}. Unfortunately, this approach suffers from a congestion problem, and can require $\Omega(n)$ time even if the cluster tree has a very small height. The reason is this: consider an edge $e$ connecting $v$ to its parent $u$, and suppose we have many---up to $\Theta(n)$---nodes in the subtree below $v$ who want to register. All of their registration messages have to pass through $e$, and this requires $\Omega(n)$ time.

\paragraph{The right approach.} We will use a different approach, which is close to the variant described in \cite{Awerbuch1992AdaptingNetworks}. Although, we will need to adjust a part of their algorithm and provide a highly non-trivial correctness proof. Unfortunately, no such proof or even guarantee statement is provided in \cite{Awerbuch1992AdaptingNetworks}.  Intuitively, we will collect the global OR of the registration and deregistration status in the cluster, using processes similar to convergecast and broadcast. We will do this by marking some edges as \textbf{dirty} during the registration process. Since this OR computation does not occur momentarily, and asynchrony might change the timing of different parts of the OR computation, we need careful statements for the algorithm and the proof of its guarantees. Below we delve into the details of this process.

\subsubsection{Registration}
When a node $v$ registers in the cluster $C$, its goal is to ensure that all the edges on the path from $v$ to the cluster root $r$ are marked as dirty. Intuitively, this will be by starting a wave of messages that go toward the root and come back to the node. The precise description needs more care. Let us call a node $u$ \textbf{finished} only once it knows that all the edges from $u$ to $r$ are marked as dirty. Therefore, the registration of the node $v$ once it becomes finished.

The registration of the node $v$ is a recursive procedure, and we call it $\reg(v)$. When node $v$ wants to register, it will invoke $\reg(v)$. After it's over, it will mark itself as \textbf{registered}. Now, let us get to the details of the procedure $\reg$: suppose that we invoked it in a node $u$.

If $u$ already is finished (either because it's a root $r$, which is always finished, or because it became finished during the invocation of this procedure for some other node), then $\reg(u)$ terminates immediately. Otherwise, let $par$ denote the parent of $u$ in the cluster tree. If the edge $(par, u)$ is not dirty, $u$ will mark it as dirty (and notify $par$ of this, sending the corresponding message). Then, it will invoke $\reg(par)$ by sending a message to $par$, and wait until $par$ notifies it that $\reg(par)$ is over. After $par$ notifies $u$ that $\reg(par)$ is over, $u$ knows that $u$ is also finished. It then notifies any children that had invoked $u$, and then $\reg(u)$ is over. 

\subsubsection{Deregistration}

Intuitively, the goal of the deregistration of node $v$ is to ``cancel'' the effect of the registration of $v$. Again, the intuitive thing would be to send a wave of messages from $v$ toward the cluster root $r$, now erasing dirty marks along the path (until reaching a node that has another child edge that remains dirty).  The precise description needs more care. 

When node $v$ wants to deregister, it immediately marks itself as \textbf{deregistered}. Then, node $v$ starts the recursive deregistration procedure $\dereg(v)$. When $\dereg(u)$ is invoked at a node $u$, this node starts by looking at the edges from $u$ to its direct children. If any of them are dirty, $\dereg(u)$ is immediately terminated. If $u = r$, or if $u$ is still registered, then $\dereg(u)$ is also immediately terminated. Otherwise, let $par$ denote the parent of $u$ in the cluster tree. As we are deregistering some descendants of $u$, the entire path from $u$ to the root is dirty, and, in particular, the edge $(par, u)$ is dirty. Then, $u$ does the following three actions:

\begin{itemize}

\item Remove the dirty mark from the edge $(par, u)$, and mark it as \textbf{waiting} instead.

Intuitively, a waiting edge indicates that a $Go\_Ahead$ has to be sent through it. 

\item Remove the finished mark from node $u$ (because edge $(par, u)$ is not dirty anymore).

\item Inform $par$ to start $\dereg(par)$.

\end{itemize}

\subsubsection{Sending {$Go\_Aheads$}}
% \todo{nodes waiting?}
As soon as the cluster root $r$ receives an update that the edge from $r$ to one of its children has had its dirty mark removed, it does this: if now all the edges from $r$ to its children are not dirty, root $r$ issues a $Go\_Ahead$. It will start propagating this $Go\_Ahead$ through the waiting edges with a recursive procedure $\go$, by invoking $\go(r)$. 

Let the current node in which $\go$ is invoked be $u$. If $u$ has deregistered, but it has not received a $\go$ yet, once it receives the $\go$, then node $u$ is immediately marked as \textbf{free}: Here, \textbf{free} means that node $u$ has received the $Go\_Ahead$ that it was waiting for and can proceed to the next stages of the algorithm. 

Any node $u$ in which $\go$ is invoked also may need to propagate the $Go\_Ahead$ to some of its children. Let $ch$ be some particular child of $u$. Node $u$ propagates $Go\_Ahead$ to child $ch$ only if the edge $(u, ch)$ is waiting. If the edge $(u, ch)$ is not waiting (e.g., if it stopped being waiting before $u$ has propagated the $Go\_Ahead$ to $ch$), then $u$ will not propagate this $Go\_Ahead$ to its child $ch$. When child $ch$ receives this propagating message, it will invoke $\go(ch)$.

\subsubsection{Analysis}
\label{subsec:registrationAnalysis}
% \todo{Maybe move proofs to Appendix?}

In this subsection, we prove that the register guarantees mentioned earlier hold, and bound the time and message complexities of registration, deregistration, and sending $Go\_Aheads$.

\begin{lemma}
    \label{lemma:monotonicity}
    If $ch$ is a child of $v$ in the cluster tree, and $ch$ is finished, then $v$ is also finished. 
\end{lemma}

\begin{proof} 
    Let us look at a maximal contiguous interval of time such that $ch$ remains finished during this entire time interval. Node $ch$ can become finished only after getting to know that $v$ is finished. Now, we have to show that $v$ remains finished until the finished mark of $ch$ is removed. Notice that the only way node $v$ might stop being finished is if $\dereg(v)$ is invoked. However, since the edge $(v, ch)$ remains dirty, $\dereg(v)$ will be immediately terminated. 
\end{proof}

\begin{lemma}
    \label{lemma:path}
    From the moment when $v$ is registered, till the moment when $v$ is deregistered, all the nodes on the path from $v$ to $r$ are finished, and all the edges on the path from $v$ to $r$ are dirty.
\end{lemma}

\begin{proof}
    Note that $v$ remains finished during this entire time interval: it is finished when it just registered, and it can stop being finished only during $\dereg(v)$, which will be immediately terminated if $v$ is still registered. Then, by the \cref{lemma:monotonicity}, the parent of $v$ will also remain finished during this time interval, as well as the parent of the parent of $v$, and so on: all the nodes on the path from $v$ to $r$ will remain finished. As all nodes on this path are finished during the entire interval, all the edges on this path remain dirty during this interval.
\end{proof}

\begin{lemma}
    \label{lemma:notalldirty}
    Suppose that node $v$ received $Go\_Ahead$. From the moment when it was issued by the root $r$, till the moment when $v$ received it, at least one edge on the path from $v$ to $r$ is not dirty.
\end{lemma}

\begin{proof}
    Consider the way $Go\_Ahead$ is propagated from $r$ to $v$. There are only two things that it's doing during this time interval: it may sit in some node $v$, waiting until $v$ decides whether to propagate it to its child $ch$, or it may be propagated from $v$ to one its children $ch$ through the edge $(v, ch)$. In both cases, the edge $(v, ch)$ has to remain not dirty: otherwise, $v$ will simply not propagate the $Go\_Ahead$ to $ch$.
\end{proof}

\begin{lemma}
    \label{lemma:issuing}
    If node $v$ receives $Go\_Ahead$, this $Go\_Ahead$ was issued by $r$ after $v$ has deregistered.
\end{lemma}

\begin{proof}
    Suppose that it was issued by $r$ before node $v$ has deregistered. By definition, it can reach $v$ only after $v$ has deregistered. Consider the point in time right before $v$ has deregistered, but after it has registered, and after this $Go\_Ahead$ has already been issued by $r$. By \cref{lemma:notalldirty}, at this point in time, at least one edge on the path from $v$ to $r$ had to be not dirty, while by \cref{lemma:path}, all these edges had to be dirty. Contradiction.
\end{proof}

\registerGuaranteeOne*

\begin{proof}[Proof of \Cref{lemma:guarantee1}]
    Suppose that node $v$ has received $Go\_Ahead$, but node $u$ hasn't deregistered yet, and that $u$ has registered before $v$ has deregistered. Let us look at the moment when this $Go\_Ahead$ was issued by the $r$. This could not happen after $u$ has registered, as the entire path from $u$ to the root is marked dirty until $u$ deregisters, by \cref{lemma:path}. So, this $Go\_Ahead$ was issued before $u$ registered and, therefore, before $v$ has deregistered, contradicting \cref{lemma:issuing}.
\end{proof}

% \mtodo{Check: changed}

\registerGuaranteeTwo*

\begin{proof}[Proof of \Cref{lemma:guarantee2}]
    Let us show that every node will eventually receive $Go\_Ahead$. Consider a moment when all the nodes that have registered have deregistered, after which no new registrations will happen, and all the procedures are over: basically, the moment after which there will not be any more updates to the status of the nodes and edges of the cluster. Suppose that node $v$ has registered and deregistered but has not received the $Go\_Ahead$. 

    Let us say that $v$ is in a \textbf{good} state, if one of the following holds:

    \begin{itemize}
        \item 
        All the edges on the path from $v$ to $r$ are marked as waiting.

        \item
        There is some node $v_1 \neq r$ on the path from $v$ to $r$, such that all the edges on the path from $v$ to $v_1$ are marked as waiting, and the edge from $v_1$ to its parent is marked as dirty.
    \end{itemize}

    Right before $v$ has deregistered, it is in a good state. %We will show that there will be points in the arbitrarily distant future when $v$ is in a good state, and show that this implies a contradiction.
    How can $v$ stop being in a good state? At each point in time, consider the higher ancestor $h$ of $v$ such that all the edges of the entire path from $v$ to $h$ are currently waiting. If $v$ is in a good state, then $h$ is the root of the cluster, or the edge from $h$ to its parent is dirty.

    This state can be broken in one of the following ways. We analyze all of them and show that $v$ will inevitably end up in a good state again after some time.

    \begin{enumerate}
        \item $Go\_Ahead$ is propagated down on the path from $h$ to $v$. By our assumption, it does not reach $v$, so it must get stuck somewhere on the path, which happens when it detects a dirty edge, which once again moves $v$ to the good state of the second type.
        
        \item Some edge on the path from $v$ to $h$ becomes dirty during the registration process of some node. Then $v$ immediately ends up in a good state of the second type.
        
        \item The edge from $h$ to its parent stops being dirty during the deregistration of some descendants of $h$ and becomes waiting instead. Let us follow this deregistration process. If it reaches the root, then $v$ ends up in a good state of the first type. Otherwise, it reaches a dirty edge, and $v$ ends up in a good state of the second type.
    \end{enumerate}

    Now, consider the last time any dirty edge has been marked as waiting. It has happened in the last $\dereg$ call, which could be only from some child of $r$. By the invariant, it follows that all the edges on the path from $v$ to $r$ have to be waiting from that moment on. But then the $Go\_Ahead$ issued by $r$ after marking this edge as waiting will reach $v$, a contradiction. 
    
    From the above proof, it also follows this $Go\_Ahead$ will reach $v$ in time at most $O(h)$. To show that the total number of messages spent on sending $Go\_Ahead$s is proportional to the total number of messages spent on registration and deregistration, note that every $Go\_Ahead$ message is removing a waiting mark from an edge.
 \end{proof}

\section{Asynchronous multi-source BFS}
\label{sec:BFS}
The (multi-source) BFS problem in the synchronized setting takes $O(D_1)$ rounds and $O(m)$ messages, where $D_1$ is the largest value of $dist(v, S)$ over all nodes $v$. Here, $S$ denotes the set of all sources, and we want each node to learn its distance only to the closest source. In this section, we develop an asynchronous BFS algorithm with complexities similar to the synchronous variant. In \Cref{sec:generalization}, we explain how we extend this scheme to a general synchronizer for an arbitrary event-driven synchronous algorithm, thus proving (the formal variant of) \Cref{thm:mainInformal}.

\begin{theorem}
There exists an asynchronous deterministic algorithm that computes single-source BFSs in $\tilde{O}(D)$ time and using $\tilde{O}(m)$ messages. 

In the case of multiple-source BFS,  there exists an asynchronous deterministic algorithm that computes multiple-source BFS in $\tilde{O}(D_1)$ time and using $\tilde{O}(m)$ messages, where $D_1$ is the largest value of $dist(v, S)$ over all nodes $v$.
\end{theorem}

As a concrete step toward this theorem, we define a thresholded version of the BFS where we want to learn the distance only for nodes whose distance to the sources is upper bounded by this threshold. The formal definition is as follows.
\begin{definition} For any nonnegative integer $\tau$, we define the  \textbf{$\tau$-thresholded BFS} problem as follows: Let $S$ be the set of source nodes. For each node $v$, if $dist(v, S)\leq \tau$, then $v$ should output $dist(v, S)$. If $dist(v, S)> \tau$, then $v$ should output a special symbol $\infty$, which indicates that $dist(v, S)> \tau$ and thus a BFS of depth $\tau$ from $S$ does not reach $v$. 
\end{definition}

\paragraph{Roadmap.} We build our BFS algorithm step by step, and over different subsections.

In \Cref{subsec:singleBFS}, we show how to run a single-source $2^t$-thresholded BFS, assuming that a layered sparse $2^{t+6}$-cover has been provided.  In \Cref{subsec:multiBFS}, we show how to extend the approach from \Cref{subsec:singleBFS} to the multi-source BFS problem. In \Cref{subsec:largerMultiBFS}, we show how to run $2^t\cdot \ell$-thresholded BFS, even for multiple sources, for some integer $\ell$. This still assumes that we are given a layered sparse $2^{t+6}$-cover, and we incur only an extra $\ell$-multiplicative increase in complexities.

The three subsections mentioned above assume that we are already given a layered sparse $2^{t+6}$-cover. We build those recursively. Concretely, in \Cref{subsec:deterministiccover} and \Cref{subsec:asynchSparseCover}, we show that if we are given a layered sparse $2^{t+6}$-cover, we can construct sparse $2^{t+7}$-cover.  \Cref{subsec:deterministiccover} is mostly a recap on the synchronous algorithm, and \Cref{subsec:asynchSparseCover} describes the adaptation to the asynchronous environment. The key ingredient in the latter reduces to the task of running several multi-source BFSs in the asynchronous environment, for which we can use the approach developed in \Cref{subsec:singleBFS} to \Cref{subsec:largerMultiBFS}. Hence, the entire construction works recursively, and the construction for each sparse cover depends on covers of lower radii. 

Finally, in \Cref{subsec:completeBFS}, we show how to construct a layered sparse $2^{O(1)}$-covers and complete the algorithm: having the layered sparse cover, we build the next sparse covers one by one, based on the already built ones. There, we also discuss the complete BFS algorithm and its termination.

We remark that our algorithm---especially much of the terminology and design choices in \Cref{subsec:singleBFS}---builds on the approach of Awerbuch and Peleg\cite{Awerbuch1990NetworkOverhead} and \cite{Awerbuch1992AdaptingNetworks}. In their case, they could rely on some nice properties of the randomized sparse cover constructions of Linial and Saks~\cite{Linial1993LowDecompositions}. Unfortunately, we do not see a black-box replacement of this randomized construction with the deterministic synchronous construction of Rozhon, Ghaffari~\cite{Rozhon2020Polylogarithmic-timeDerandomization}. However, we arrange the overall algorithm such that the key asynchronous task in building sparse covers boils down to an asynchronous multi-source BFS computation, and we solve that recursively by leveraging lower radii sparse covers. Furthermore, as mentioned before, we need to adapt their approach in some parts and provide some highly nontrivial correctness and complexity analysis.

\subsection{Single-source $2^{t}$-thresholded BFS, given a layered sparse $2^{t+6}$-cover}
\label{subsec:singleBFS}
Here, we describe a polylogarithmic synchronizer for single-source $2^{t}$-thresholded BFS, assuming that we are given $2^i$-covers for $i$ from $1$ to $t+6$. In the rest of this subsection, we denote the source of this BFS by $s$.

\subsubsection{Basic definitions}

In the synchronous environment, a single-source BFS is implemented as follows: at pulse $p$, nodes at a distance $p$ from the source send the ``join'' proposal to all their neighbors. Node $v$ accepts only the first ``join'' proposal it receives. If it receives several of them at the same pulse, it accepts an arbitrary one. If $v$ received this proposal at the end of the round $p$, then $v$ determines that it is at distance $p$ from the source.

Now, let us move to the asynchronous world. Let $pulse(v)$ denote the \textbf{pulse} of node $v$, as we define next (we will prove, as an ingredient of the algorithm's correctness, that these pulse values are indeed equal to the distance of the node from the source). In particular, $pulse(s) = 0$. As in the synchronous case, nodes will send the ``join'' proposals. A node will accept only the first ``join'' proposal that it receives. Each node will respond to each ``join'' proposal it receives, indicating whether it accepts or declines. We define \textbf{$parent(v)$} as a node from which $v$ received the first ``join'' proposal, i.e., the proposal which it accepted. The graph formed on edges $(v, parent(v))$ will form a tree, to which we will refer as the \textbf{execution tree}.

A node will set its own $pulse(v)$ to $pulse(parent(v)) + 1$. Our goal is devise the algorithm so that $pulse(v)$ will indeed be equal to $dist(v, s)$. That is, we have to ensure that the first ``join'' proposal that the node at a distance $p+1$ from $s$ will receive will be from a node with pulse $p$. We will prove that our algorithm satisfies this constraint in the \cref{lemma:executionisbfs}.

\begin{definition} The \textbf{level} $\ell(p)$ of the pulse $p$ is $\infty$ for $p = 0$, and otherwise equal to $i$ such that there exists an integer $j$ satisfying $p = (2j+1)2^i$. That is, the logarithm of the highest power of $2$ that divides $p$.
\end{definition}

\begin{definition}
    For any value $p>0$, let $\tilde{p}$ be the largest number satisfying the following conditions:
    
    \begin{itemize}
        \item $\ell(\tilde{p}) = \ell(p)+1$
        \item $\tilde{p} \le p - 2^{\ell(p)}.$
    \end{itemize}
    The pulse $prev(p)$ for a pulse $p$ is defined as $max(\tilde{p}, 0)$. And we define $prev(0) = 0$.
\end{definition}

\begin{definition}
      The \textbf{host} $h(v)$ of a node $v$ is defined as the maximum-depth ancestor $h$ of $v$ in the execution tree that satisfies $pulse(h) = prev(pulse(v))$. In particular, $h(s) = s$.
\end{definition}

\subsubsection{The algorithm} Intuitively, a node of pulse $p$ should send its join proposals only after it has received all the join proposals from nodes with pulses up to $p-1$. Toward this, we make $host(v)$ and $host(host(v))$ gather some safety information, as we describe in the next definition.

\begin{definition}
    For a node $v$, let us call it \textbf{$p$-empty} if there are no nodes of pulse $p$ in its subtree in the execution tree. Let us call it \textbf{$p$-safe} if it is $p$-empty or if all nodes of pulse up to $p-1$ in its subtree have already sent their ``join'' proposals and received answers to all of them (declines/accepts).
\end{definition}

For each pulse $p>0$, and for every node $v$ with $prev(prev(p)) \le pulse(v) \le p$, we want to know if it is $p$-empty, and if not, we want to know when it becomes $p$-safe. One can gather this information using a convergecast along the already built partial execution tree: each node simply sends the information to its parent when it receives it. For each pulse $p$, we propagate the information about nodes being $p$-safe and $p$-empty only up to the first node of level $prev(prev(p))$.

How do we use this information? For every pulse $p$ up to $2^t$, for every node $v$ with $pulse(v) = prev(prev(p))$, we do the following:

\begin{itemize}
    \item If $v$ is not $prev(p)$-empty, and it got to know that it is  $prev(p)$-safe, then $v$ will \textbf{$p$-register} in all clusters of the $2^{\ell(p)+5}$-cover that contain $v$, using the registration abstraction from the \cref{subsec:registration}. Then, node $v$ is  called \textbf{$p$-registered}.
    
    Once node $v$ receives confirmation that it is  $p$-registered, node $v$ passes the information about $prev(p)$-safety to its parent (this information will be passed up to the pulse $prev(prev(p))$).

    \item If $pulse(v) = prev(prev(p))$, then, if $v$ is not $prev(p)$-empty but $p$-safe, then it \textbf{$p$-deregisters} in all clusters of the $2^{\ell(p)+5}$-cover that contain $v$. Then, node $v$ is called \textbf{$p$-deregistered}.

    After that, node $v$ waits until it receives a $Go\_Ahead(p)$ from all these clusters.
\end{itemize}

For a single cluster, the registration for different pulses are disjoint. In the terminology of the \cref{lemma:sequential}, we imagine $p$-registrations/$p$-deregistrations/sending $Go\_Ahead(p)$ as different stages of the algorithm, for different $p$. It is possible for a node $v$ to have to register/deregister for more than one pulse: $v$ does it for all $p$, for which $pulse(v) = prev(prev(p))$ (and $v$ is not $prev(p)$-empty).

When node $v$ receives $Go\_Ahead(p)$, it propagates this $Go\_Ahead$ down the execution tree to all its children with pulse $p$. When node $v$ of pulse $p$ receives such a $Go\_Ahead(p)$, it starts sending the messages to its neighbors.
This is done for all pulses $p$ up to $2^t-1$ (so that nodes at distance $2^t$ from the source learn their pulse, but do not send ``join'' proposals to their neighbors).

\paragraph{Checking stage.} At the end of this process, all nodes at a distance at most $2^t$ from $s$ will learn their pulses. However, nodes at larger distances will not learn that they are at a distance larger than $2^t$ from $s$: they will keep waiting. Notice that in the thresholded BFS, we would like these nodes to output $\infty$. To solve this issue, we make one last addition to the algorithm. 

Let us use an abstraction from \cref{SUBSEC:BALLCOLLECTING}. Let process $\mathcal{P}$ denote ``being the source $s$ and becoming $2^t$-safe'', so that node $v$ is done with $\mathcal{P}$ when it becomes $2^t$-safe if $v = s$, or at the very beginning, if $v \neq s$. Let each node learn when its $2^t$-neighborhood is done with $\mathcal{P}$. When node $v$ learns that this, then, if it was reached by the BFS, it knows its pulse. Otherwise, it learns that $dist(v, s) > 2^t$. Indeed, if $dist(v, s) \le 2^t$, then $v$ and $s$ are in some cluster of $2^t$-cover together, but then $v$ must have learned its pulse before $s$ becomes $2^t$-safe and $v$ learns that its $2^t$-neighborhood is done with $\mathcal{P}$.

\subsubsection{Analysis}

First, in \Cref{lemma:closehosts}, we show that if two nodes of the same pulse are ``nearby'', then their hosts are also pretty close to each other. We first need some more basic facts.

\begin{lemma}
    \label{lemma:prevclose}
    For any pulse $p>0$, the following holds: \textbf{(a)} $p - prev(p) \le 3\cdot 2^{\ell(p)}$. \textbf{(b) }$p - prev(prev(p)) \le 9\cdot 2^{\ell(p)}$. Furthermore, this implies that we have the following for any node $v$ with pulse $p>0$: \textbf{(c)} $dist(p, h(v)) \le 3\cdot 2^{\ell(p)}$. \textbf{(d)} $dist(p, h(h(v))) \le 9\cdot 2^{\ell(p)}$.
\end{lemma}

\begin{proof}
    For the first claim, numbers $p - 2^{\ell(p)}$ and $p - 3\cdot 2^{\ell(p)}$ are divisible by $2^{\ell(p)+1}$, but only one of them can be divisible by $2^{\ell(p)+2}$, so $prev(p) \ge p - 3\cdot 2^{\ell(p)}$.

    For the second claim, if $prev(p) = 0$, then $prev(prev(p)) = 0 = prev(p) \ge p - 3\cdot 2^{\ell(p)} > p - 9\cdot 2^{\ell(p)}$. Otherwise, $\ell(pulse(p)) = l(p) + 1$, and $prev(p) - prev(prev(p)) \le 3 \cdot 2^{\ell(p)+1}$. As $p - prev(p) \le 3 \cdot 2^{\ell(p)}$, we can add these inequalities to obtain $p - prev(prev(p)) \le 9 \cdot 2^{\ell(p)}$. 
\end{proof}

\begin{lemma}
\label{lemma:closehosts}
    For any two nodes $u$ and $v$ of pulse $p>0$ in the execution tree such that $dist(u, v) \le 2^{\ell(p)+3}$ in the graph, we have that $dist(h(u), h(v)) \le 2^{\ell(p)+4}$.    
\end{lemma}

\begin{proof}
    We have $dist(h(u), h(v)) \le dist(h(u), u) + dist(h(v), v) + dist(u, v) \le 3\cdot 2^{\ell(p)} + 3\cdot 2^{\ell(p)}  + 2^{\ell(p)+3} = 14 \cdot 2^{\ell(p)} < 2^{\ell(p)+4}$.
\end{proof}

 Now, we get to the core lemma of the correctness proof.
 
\begin{lemma}
    \label{lemma:waiting}    
    For a pulse $p>0$ consider any two nodes $v_1, v_2$ with pulse $prev(prev(p))$, such that $dist(v_1, v_2) \le 2^{\ell(p) + 5}$, $v_1$ is not $p$-empty, and $v_2$ is not $prev(p)$-empty. Then, for any cluster $C$ of the $2^{\ell(p)+5}$-cover that includes both of $v_1$ and $v_2$, the following holds:
    \begin{itemize}
        \item $v_1$ will $p$-deregister in $C$ only after $v_2$ $p$-registers in $C$.
        \item $v_1$ will receive $Go\_Ahead(p)$ in $C$ only after $v_2$ $p$-deregisters in $C$.
    \end{itemize}
\end{lemma}

\begin{proof}
    The proof is by induction based on the pulse number. If $prev(prev(p)) = 0$, then there is only one node with pulse $prev(prev(p))$: the source $s$. Then $v_1 = v_2 = s$, and the statement holds. 
    
    Consider the case $prev(prev(p)) > 0$, and suppose that the statement holds for all smaller values of $p$. Consider any such cluster $C$. Suppose that $v_1$ is about to $p$-deregister in it. Consider the hosts of $v_1, v_2$: $h(v_1), h(v_2)$. By \Cref{lemma:closehosts} we have $dist(h(v_1), h(v_2)) \le 2^{\ell(p) + 6}$. It follows that there is some cluster $C_1$ of $2^{\ell(p)+6}$-cover, in which both $h(v_1)$ and $h(v_2)$ are present.

    Let us start by proving the first claim. As $v_1$ has already $p$-deregistered in $C$, some nodes of pulse $\ge prev(p)$ in the subtree of $v_1$ have already sent their messages. This is possible only after $h(v_1)$ has received $Go\_Ahead(prev(p))$ in $C_1$, which, by the induction hypothesis, is possible only after $h(v_2)$ $prev(p)$-deregisters in $C_1$. This is possible only after $v_2$ $p$-registers in $C$.

    Now the second claim. Consider the moment when $v_1$ is about to receive $Go\_Ahead(p)$. This means that $v_1$ has already $p$-deregistered. By the first claim, $v_2$ has $p$-registered before that. By \cref{lemma:guarantee1}, as $v_2$ has $p$-registered before $v_1$ has $p$-deregistered, $v_1$ will receive $Go\_Ahead(p)$ only after $v_2$ $p$-deregisters. 
\end{proof}

    Note that by now, we have only used that $pulse(v) = pulse(parent(v)) + 1$. Now, we will show that our execution tree is actually a BFS tree.

\begin{lemma}
    \label{lemma:executionisbfs} The above algorithm computes a $2^{t}$-thresholded BFS: For every node $v$ such that $dist(v, s)\leq 2^{t}$, we have $pulse(v) = dist(v, s)$. In other words, this execution tree is indeed a BFS tree. And for any node $v$ such that $dist(v, s)> 2^{t}$, the node learns that its distance is greater than $2^{t}$.
\end{lemma}

\begin{proof}
    The second part of the lemma statements is guaranteed by the checking stage. Here, we focus on the first part. 
    We will prove by induction by pulse $p$ the following statement: for every node $v$ with $dist(v, s) = p$ holds $pulse(v) = p$. This statement is trivially true for $p = 0$.

    Suppose that this is proved for all pulses up to $p$; we prove it for pulse $p+1$.
    Consider any node $v$ with distance $p+1$ to the root. We show that the first ``join'' proposal it receives will be from a node with pulse $p$. Suppose the contrary. Let $v_1$ be a neighbor with $dist(v_1, s) = p$, and suppose that $v$ received the first ``join'' proposal from some node $v_2$ with $pulse(v_2) = p_1 > p$. Consider the moment when $v$ has already received a ``join'' proposal from $v_2$ but has not yet received it from $v_1$.
    Let $x$ be the integer from the segment $[p+1, p_1]$ divisible by the largest power of $2$ (with largest $\ell(x)$), and let $w$ be the ancestor of $v_2$ with pulse $x$. As $v$ was reached in a subtree of $s$, $w$ has already sent at least some messages of pulse $x$. This is possible only after $w$ receives $Go\_Ahead(x)$, which is possible only after $h(h(w))$ receives $Go\_Ahead(x)$.

    Now note that $pulse(h(w)) \le p$. Indeed, $pulse(h(w)) \le w - 2^{\ell(w)}$, and as number $x - 2^{\ell(x)}$ is divisible by larger power of $2$ than $x$ is, it cannot be in $[p+1, p_1]$ by our choice of $x$. So, $pulse(h(w)) \le x - 2^{\ell(x)} \le p$. Let $y_1$ be the ancestor of $v_1$ with pulse $prev(prev(x))$, $y_2$ be the ancestor of $v_2$ with pulse $prev(prev(x))$ (so that $y_2 = 
    h(w)$). As $pulse(h(w)) \le p$, $y_1$ is not $prev(x)$-empty. $y_2$ is not $x$-empty by the choice of $x$.
    Consider the distance between $y_1, y_2$. $dist(y_1, y_2) \le dist(y_1, v) + dist(v, y_2) \le (1 + (p - prev(prev(x)))) + (1 + (p_1 - prev(prev(x)))) = 2 + 2(x - prev(prev(x))) + (p - x) + (p_1 - x)$. By \cref{lemma:prevclose}, $x - prev(prev(x)) \le 9 \cdot 2^{\ell(p)}$. $p-x \le -1$ as $x \in [p+1, p_1]$, and $p_1-x \le 2^{\ell(p)}-1$ as otherwise we could have chosen $x + 2^{\ell(p)}$, which is divisible by higher power of $2$. So, $dist(y_1, y_2) \le 2 + 18\cdot 2^{\ell(p)} + 2^{\ell(p)} - 2 = 19 \cdot 2^{\ell(p)} \le 2^{\ell(p) + 5}$. Then, by \cref{lemma:waiting}, $y_2 = h(h(w))$ can receive $Go\_Ahead(x)$ only after $y_1$ has $x$-deregistered, which is not possible until $v_1$ has sent its messages of pulse $p$, contradiction.
\end{proof}

\begin{theorem}
\label{theorem:SSBFStimeandmessage}
    The $2^t$-thresholded single-source BFS given a layered sparse $2^{t+6}$-cover finishes in time $O(2^t\cdot \log^8{n})$ and uses message complexity $O(m\log^5{n})$.
\end{theorem}

We discuss time and message complexity parts separately.

\paragraph{Analysis of the time complexity.}
Next, we analyze the time complexity of our single-source BFS algorithm and prove the time complexity part of \Cref{theorem:SSBFStimeandmessage}. We first discuss some helper lemmas.

We divide all our operations into stages, one stage $\stage_p$ for each pulse $p$. To $\stage_p$ we include all the $p$-registrations, $p$-deregistrations, sending $Go\_Ahead(p)$s, collecting information about $p$-safety, and sending ``join'' proposals from nodes of pulse $p-1$ to nodes of pulses $p$: basically, everything related to pulse $p$. It's clear that each stage depends only on the previous pulses, so the runtime of the algorithm is bounded by the isolated time complexities of each stage.

\begin{lemma}
    \label{lemma:stagebound}
    The isolated time complexity of $\stage_p$ is $O(2^{\ell(p)} \log^7{n})$.
\end{lemma}

\begin{proof}
    Each edge is in $O(\log^4{n})$ clusters of the sparse $2^{\ell(p)+5}$-cover. In the execution tree, each edge is involved in collecting the information about $p$-safety only $O(1)$ times. So, we may say that each edge is involved in $O(\log^4{n})$ separate tasks for $\stage_p$, and multiply the runtime by this factor in the end, by \cref{corollary:parallel}.

    After all nodes of pulses at most $p-1$ have sent their ``join'' proposals, all of them will be acknowledged in $O(1)$ time units. All nodes of pulse $prev(p)$ will receive information about their $p$-safety/$p$-emptiness in at most $3\cdot 2^{\ell(p)}$ time units after that. All these nodes (which are not $p$-empty) will $p$-register in the corresponding clusters in at most $O(2^{\ell(p)}\cdot \log^3{n})$ times units. All nodes of pulse $prev(prev(p))$ will receive information about their $p$-safety in at most $6\cdot 2^{\ell(p)}$ time units after that. All these nodes will $p$-deregister in the corresponding clusters in at most $O(2^{\ell(p)}\cdot \log^3{n}) $ times units. 

    By \cref{lemma:guarantee2}, all clusters of the $2^{\ell(p)+5}$-cover will then send $Go\_Ahead(p)$s, which will be received in $O(2^{\ell(p)}\cdot \log^3{n})$ time units. They will then be propagated to all the nodes of pulse $p$ in at most $9\cdot 2^{\ell(p)}$ after that. The total time complexity of this is $O(2^{\ell(p)}\cdot \log^3{n})$

    Taking into the account that there are $O(\log^4{n})$ tasks for each edge, the isolated time complexity of $\stage_p$ becomes $O(2^{\ell(p)} \cdot \log^7{n})$.
\end{proof}

\begin{lemma}
    \label{lemma:summation}
    For any integer $t$, $\sum_{p = 1}^{2^t} 2^{\ell(p)} = O(2^t \cdot t)$.
\end{lemma}

\begin{proof}
    The number of values $p$ such that $\ell(p) = k$ is $\le 2^{t-k}$. Hence, $\sum_{p = 1}^{2^t} 2^{\ell(p)} \leq \sum_{k = 0}^t 2^k\cdot 2^{t-k} = (t+1)2^t =  O(2^t \cdot t)$.
\end{proof}

\begin{proof}[Proof of the time complexity part of \Cref{theorem:SSBFStimeandmessage}]

   The total isolated time complexity of of all stages $\stage_i$ becomes $\sum_{p = 1}^{2^t} O(2^{\ell(p)}\cdot \log^7{n}) = O(2^t \cdot \log^8{n})$. The isolated time complexity of the checking stage is, by \cref{theorem:convergecast}, $O(2^t\log^7{n})$. So, the entire runtime is $O(2^t \cdot \log^8{n})$. 
\end{proof}

\paragraph{Analysis of the message complexity.}
Next, we analyze the message complexity of our single-source BFS algorithm and provide the proof for the message complexity part of \Cref{theorem:SSBFStimeandmessage}. First, we need to prove a small lemma.

\begin{lemma}
    \label{lemma:pulsesperedge}
    For any $p_1$, there exist $O(t)$ pulses $p \le 2^t$ such that $prev(prev(p)) \le p_1 \le p$.
\end{lemma}

\begin{proof}
    It is  enough to show that there are $O(1)$ pulses of each particular level. For a pulse $p$ of given level $\ell$, we know $prev(prev(p)) \ge p - 9\cdot 2^{\ell}$. So, for any such pulse $p$, we must have $p - 9\cdot 2^{\ell} \le p_1 \le p$. There can be at most $10$ such pulses.
\end{proof}

\begin{proof}[Proof of the message complexity part of \Cref{theorem:SSBFStimeandmessage}]
    The number of messages sent for building the execution tree is $O(m)$. The number of messages propagated in the execution tree is $O(nt) = O(n\log{n})$, as, by \cref{lemma:pulsesperedge}, there are messages of only $O(t)$ pulses propagated through each edge. Now, consider messages in the clusters of the sparse covers.

    Consider some level $\ell$. For a pulse $p$ with $\ell(p) = \ell$, node $v$ with pulse $prev(prev(p))$ will $p$-register in the clusters of the sparse $2^{l + 5}$-cover only if it has some children of pulse $prev(p)$, and therefore has at least $2^{l}$ nodes on the path to such a child. As these $2^{l}$ nodes are unique for each such node $v$, it follows that for a given level $\ell$, the total number of nodes who will $p$-register for some pulse $p$ with $\ell(p) = \ell$ is $O(\frac{n}{2^{\ell}})$. Such nodes will register in clusters of $2^{\ell+5}$ covers and therefore take $O(2^{\ell}\log^3{n})$ messages to the root and back. Every node is in $O(\log{n})$ clusters, so, there will be $O(2^{\ell}\log^3{n} \cdot \log{n} \cdot \frac{n}{2^{\ell}}) = O(n\log^4{n})$ messages corresponding to registration at a given level, and the total number for all levels is $O(n \log^5{n})$. We assume that the graph is connected, so $n = O(m)$, and thus message complexity is $O(m\log^5{n})$. 

    Finally, the number of messages used in the checking stage is $O(m\log^4{n})$. So, the total number of messages is $O(m\log^5{n})$.
\end{proof}

\subsection{Multi-source $2^{t}$-thresholded BFS, given a layered sparse $2^{t+6}$-cover}
\label{subsec:multiBFS}
Here, we describe a multi-source variant of \Cref{subsec:singleBFS}. The entire algorithm works the same way as before, and all the proofs remain valid, except for one critical issue: the base case of the \cref{lemma:waiting}. Notice that \cref{lemma:waiting} has to hold even if $prev(prev(p)) = 0$ (that is, for the case when both $v_1, v_2$ are sources). In the case of a single source BFS there was only one source, so the statement was true for the base case automatically. Now we have to take care of this algorithmically.  

We will make the statement hold for sources by changing the registration procedure for them. Consider any pulse $0 < p \le 2^t$ with $prev(prev(p)) = 0$ and any cluster $C$ of $2^{\ell(p)+5}$ cover. 

\paragraph{Registration.} We will make sure that all sources in $C$ have $p$-registered, with a convergecast (note that we do not require them to be $prev(p)$-safe). After the root of the cluster learns that all its sources have been $p$-registered, it propagates this information down to all the sources. A source may send its messages only after it has registered for all pulses $p$ with $prev(prev(p)) = 0$ in all clusters of $2^{\ell(p)+5}$ covers and received a confirmation from all of them.

\paragraph{Deregistration.} After source learns that it is $p$-safe, it will $p$-deregister. This $p$-deregistration will be collected in the cluster in the same manner, with a simple convergecast.

\paragraph{Sending $Go\_Ahead$.} After the root of the cluster learns that all the sources in it have $p$-deregistered, it propagates down the $Go\_Ahead(p)$s. This way, the condition of the lemma is easily satisfied: for every cluster, all the registrations happen before any deregistration is even possible.

\begin{theorem}
\label{theorem:MSBFS}
    The $2^t$-thresholded multi-source BFS given a layered sparse $2^{t+6}$-cover finishes in time $O(2^t\cdot \log^8{n})$ and takes $O(m\log^5{n})$ messages.
\end{theorem}

Before providing the proof of this statement, we discuss a helper lemma. 

\begin{lemma}
    \label{lem:NumberofPrevPrev}

    The number of pulses $p$ with $0 < p \le 2^t$ and $prev(prev(p)) = 0$ is $O(t)$.
    
\end{lemma}

\begin{proof}
    We show that for each level $\ell \le t$, there are only $O(1)$ pulses $p$ with $\ell(p) = \ell$ with $prev(prev(p)) = 0$. Indeed, if $\ell(p) = l$, then we get $0 = prev(prev(p)) \ge p - 9\cdot 2^{\ell}$, so $p \le 9\cdot 2^{\ell}$, there are only $O(1)$ such $p$.
\end{proof}

\begin{proof}[Proof of \Cref{theorem:MSBFS}]

\paragraph{Time complexity.} We can view this convergecast simply as an extra stage of the algorithm, the very first one. If we had a copy of each edge for every cluster in which it's present, this stage would run in $O(2^t\cdot \log^3{n})$. As each edge is present in $O(\log^4{n})$ clusters per cover, and we are doing this for $O(\log{n})$ pulses, the total complexity of this stage, by \cref{corollary:parallel}, is $O(2^t \cdot \log^8{n})$. As the remaining part also works in $O(2^t \cdot \log^8{n})$, the resulting complexity is $O(2^t \cdot \log^8{n})$.

\paragraph{Message complexity.} For a fixed pulse $p$, the total number of messages spent on collecting this info for all the clusters is $O(m\log^4{n})$, as each edge is in $O(\log^4{n})$ clusters.

As the total number of pulses $p$ with $prev(prev(p)) = 0$ is $O(t) = O(\log{n})$, the total message complexity of this process is $O(m\log^5{n})$. 
The message complexity of the remaining part of the algorithm remains $O(m\log^5{n})$, as in \cref{theorem:SSBFStimeandmessage}, so we get a total message complexity of $O(m\log^5{n})$.
\end{proof}

\subsection{Multi-source $2^{t}\cdot \ell$-thresholded BFS, given a layered sparse $2^{t+6}$-cover}
\label{subsec:largerMultiBFS}
Let us assume that we are given $2^i$-covers for all $i$ from $1$ to $t+6$. In \Cref{subsec:multiBFS}, we discussed how to perform a multiple-source $2^t$-thresholded BFS. Now, we discuss how we do a multi-source $2^t\cdot \ell$-thresholded BFS, where $\ell$ is some positive integer, using a layered sparse $2^{t+6}$-cover.

Let us divide our BFS into $\ell$ stages. The $T$-th of these stages will have to be a $2^t$-thresholded BFS with nodes of pulse $T \cdot 2^t$ as sources. Now, Let us use an abstraction from \cref{SUBSEC:BALLCOLLECTING}. Let process $\mathcal{P}_T$ denote ``doing the BFS up to $2^t$ as a source of the $T$-th stage'', so that node $v$ completes $\mathcal{P}_T$ when it completes this BFS or learns that it is not a source in the $T$-th stage. For each $T < \ell$, let each node learn when its $2^t$-neighborhood is done with $\mathcal{P}_T$. When node $v$ learns this, he has two options. If $v$ was reached by a BFS, then, if its pulse is $T \cdot 2^t$, then it is a source at the $T+1$-st stage; otherwise, it is not. If it has not received a pulse, then this means that there are no sources of the $T$-th stage in its $2^t$-neighborhood, as if there were some, $v$ would have received a pulse from them in the $T$-th stage. Therefore, it cannot be at a distance precisely $T\cdot 2^t$ from the sources.

\begin{theorem}
 \label{theorem:MSBFSlong}
    The $2^t\cdot \ell$-thresholded multi-source BFS given a layered sparse $2^{t+6}$-cover finishes in time $O(2^t\ell \cdot \log^8{n})$ and takes $O(m\ell \log^5{n})$ messages.
\end{theorem}     

\begin{proof}

\paragraph{Time complexity.} We have $\ell$ stages, where each stage is dependent only on the previous ones, so we can apply \cref{lemma:sequential}. The isolated time complexity of each stage is $O(n\log^8{n})$, and the isolated time complexity of the convergecast is $O(n\log^7{n})$. We get a total runtime of $O(n\ell \log^8{n})$.

\paragraph{Message complexity.} Each of $\ell$ stages requires $O(m\log^5{n})$ messages. As there are $\ell$ stages, the total message complexity is $O(m\ell \log^5{n})$.
\end{proof}

\begin{remark}
    Note that we also get an algorithm $d$-thresholded BFS for any $d \le 2^t \cdot \ell $ with the same complexity. That is, this algorithm works not only for the multiples of $2^t$.
\end{remark}

\subsection{Synchronous deterministic construction of the $d$-cover}
\label{subsec:deterministiccover}
So far, we have assumed that we are given sparse covers. In this subsection, we review the synchronous construction of Rozhon and Ghaffari~\cite{Rozhon2020Polylogarithmic-timeDerandomization}. In the next subsection, we explain how to adapt this construction to the asynchronous environment, using our asynchronous multi-source BFS. 

We will start by introducing a deterministic algorithm for constructing the $d$-cover in $O(d\cdot \polylog(n))$ rounds in a synchronous environment. For that, we will need a network decomposition.

\begin{definition}
    ($k$-separated Weak Diameter Network Decomposition) Given a graph $G=(V, E)$, we define a $(\mathcal{C}, \mathcal{D})$ $k$-separated weak diameter network decomposition to be a partition of $G$ into vertex-disjoint graphs $G_1$, $G_2$, \dots, $G_\mathcal{C}$ such that for each $i\in \{1, 2, \dots, \mathcal{C}\}$, we have the following property: the graph $G_i$ is made of a number of vertex-disjoint clusters $X_1$, $X_2$, \dots, $X_\ell$, so that:

    \begin{itemize}
        \item For any $X_a$ and any two vertices $v, u \in X_a$, $dist(u, v) \le D$ in graph $G$.
        
        \item For any two $X_a, X_b$ with $j\neq l$ and any $u \in X_a$, $v \in X_b$ holds $dist(u, v) > k$.
    \end{itemize}
\end{definition}

\subsubsection{Synchronous construction of $k$-separated Network Decomposition}

Before stating the result we will use, we recall the notion of Steiner trees. A Steiner tree of a cluster is a tree with nodes labeled as \textit{terminal} and \textit{nonterminal}; the aim is to connect all terminals, possibly using some nonterminals.

\begin{restatable}{theorem}{MAINALGGK}
\label[theorem]{theorem:main_alg_gk}[Rozhon and Ghaffari\cite{Rozhon2020Polylogarithmic-timeDerandomization}]
There is an algorithm that, given a value $k$ known to all nodes, in $O(k \log^{10} n )$ communication rounds outputs a $k$-separated weak-diameter network decomposition of $G$ with $O(\log n)$ color classes, each one with $O(k \cdot \log^3 n)$ weak-diameter in $G$.

Moreover, for each color and each cluster $\mathcal{C}$ of vertices with this color, we have a Steiner tree $T_\mathcal{C}$ with radius $O(k \cdot \log^3 n)$ in $G$, for which the set of terminal nodes is equal to $\mathcal{C}$. 
Furthermore, each edge in $G$ is in $O(\log^4{n})$ of these Steiner trees.
\end{restatable}

To remain self-contained, we provide a (simplified) description of the algorithm of \cite{Rozhon2020Polylogarithmic-timeDerandomization} in the \Cref{subapp:main_alg_gk}. Our asynchronous construction will refer to the language used in this description. For proof of the correctness of their algorithm, we refer to \cite{Rozhon2020Polylogarithmic-timeDerandomization}.

\subsubsection{Constructing $d$-cover in the synchronous setting}

\begin{theorem}
    \label{theorem:synchcover}
     There is a synchronous algorithm that, given a value $d$ that is known to all nodes, in $O(d \log^{10}{n})$ communication rounds outputs a sparse $d$-cover of this graph, together with the Steiner trees of its clusters, such that each edge appears in only $O(\log^4{n})$ cluster trees.
\end{theorem}

\begin{proof}
    We start by constructing a $2d+1$ separated weak network decomposition in $O(d \log^{10}{n})$ rounds, according to \cref{theorem:main_alg_gk}. Then, for each color separately, we expand each of its clusters to its $d$-neighborhood. As different clusters of the same color are at least $2d+1$ apart, the clusters remain disjoint, so each node will be in $O(\log{n})$ clusters (at most one per color). Additionally, each edge will join at most one cluster tree, so each edge still appears in only $O(\log^4{n})$ cluster trees.
\end{proof}

\subsection{Asynchronous construction of sparse $2^{t+7}$-cover, given a layered sparse $2^{t+6}$-cover}
\label{subsec:asynchSparseCover}
The algorithm described in \cref{subsec:deterministiccover} for constructing sparse $d$-covers consists of simple multi-sources BFSs up to distance $O(d)$. And we already know how to do that, via the approach we described in \Cref{subsec:largerMultiBFS}. 
Using this connection, below we show that we can run the sparse cover construction in the asynchronous environment, resulting in the following statement:
\begin{theorem}
    \label{theorem:sparseasynch}
    There is an asynchronous algorithm that, given a layered sparse $2^{t+6}$-cover, 

    \begin{itemize}
        \item Constructs sparse $2^{t+7}$-cover together with the Steiner trees of the clusters
        \item Lets each node learn when all the Steiner trees in which it is have been built
        \item Works in time $O(2^t\cdot \log^{11}{n})$ and takes $O(m\log^8{n})$ messages
    \end{itemize}

\end{theorem}

\begin{proof}

We start by providing an asynchronous algorithm to build $k$-separated weak decomposition from the \cref{theorem:main_alg_gk}. We build one color at a time. We let each node start to build the next color only after it is done with the previous colors. We will formalize ``being done'' with color and the way to learn this later.

Consider the algorithm for building the clusters of a particular color
from \cref{lemma:main_alg_gk}. It contains $O(\log{n})$ phases; we will synchronize each phase separately and let each node enter the next phase only after it learns that it is done with the previous phases. We will formalize ``being done'' with a phase and the way to learn this later. 

One phase, as in the synchronous algorithm, consists of $O(\log^2{n})$ steps. We will synchronize each of these steps separately and let each node enter the next step only after it is done with the previous steps. We will formalize ``being done'' with a phase and the way to learn this later.

Denote $2 \cdot 2^{t+7} + 1 = d$. One step consists of two substeps.

\begin{enumerate}
    \item $d$-thresholded BFS from all blue nodes. According to the \cref{theorem:MSBFSlong}, we can do it in an asynchronous environment in just $O(2^t\log^8{n})$ time and $O(m\log^5{n})$ messages. At the end of this BFS, by the definition of the $d$-thresholded BFS, every node will be reached or will learn that its distance from any blue node is larger than $d$. 
    \item In the same way as in the original algorithm from \cref{lemma:main_alg_gk}, the roots of all Steiner trees will collect the number of proposing red nodes and make corresponding decisions, with a convergecast by the Steiner tree. A node will propagate the information up the tree only after it is done with the previous substep.
\end{enumerate}

We say that node learns that it is done with a step as soon as it receives (and propagates down the Steiner tree) all the decisions. When a node is done with a step, it knows whether it is a source for the next stage (it knows whether it is a blue node).
We say that node learns that it is done with a phase once it learns that it is done with all its steps. We say that a node is done with a color once it learns that is done with all the phases of the algorithm for building that color.
After a node knows that it is done with building all $O(\log{n})$ colors, it starts performing a $2^{t+7}$-thresholded BFS for every cluster that it is a part of, as in \cref{theorem:synchcover}.

We now discuss time complexity. Let us bound the isolated time complexity of each step of the algorithm. Isolated time complexity of the BFS substep is $O(2^t \cdot \log^8{n})$, of the convergecast substep is $O(2^t\log^6{n})$, so the entire step takes $O(2^t\log^8{n})$.  As there are $O(\log^2{n})$ steps per phase and $\log{n}$ phases per color, we get the resulting runtime of $O(2^t \log^{11}{n})$. 

We now discuss message complexity. The message complexity of a single step is $O(m\log^5{n})$ per BFS plus $O(m\log^3{n})$ per convergecast, as each edge is in at most $\log^3{n}$ clusters for a particular color, which is $O(m\log^5{n})$ in total. As there are $O(\log^2{n})$ steps per phase and $\log{n}$ phases per color, we get the resulting message complexity of $O(m \log^{8}{n})$. 
\end{proof}

\subsection{The complete BFS algorithm in $\tilde{O}(D)$ time and $\tilde{O}(m)$ messages}
\label{subsec:completeBFS}

\begin{theorem}
    \label{theorem:SSBFSfinal}
    There is an algorithm for a single-source BFS that terminates in time $O(D \log^{11}{n})$, where $D$ is the diameter of the graph and uses $O(m\log^{10}{n})$ messages.    
\end{theorem}

\begin{proof}

The algorithm works in several steps.

\paragraph{Step 1: Initalization -- constructing layered sparse $2^6$-cover.} We learned how to construct sparse $2^{t+7}$ cover given sparse $2^i$-covers for $5 \le i \le t+6$ (note that we do not ever use smaller sparse $2^i$-covers for $i \le 4$). Now the question is, how to construct $2^5$-cover and $2^6$-cover.

For a constant $c$, the synchronous algorithm for constructing $c$-cover works in $O(\log^{10}{n})$ rounds. So, we can afford to synchronize it with a simple $\alpha$-synchronizer: just sending a message through every edge at each pulse. This takes $O(\log^{10}{n})$ time and $O(m\log^{10}{n})$ messages. Every node will wait until it learns that it is done with this step before proceeding to the next one. See also \Cref{app:alpha-beta-gamma} for a more general description of the $\alpha$ synchronizer.

\paragraph{Step 2: One iteration.} Every iteration will consist of two parts.

\begin{enumerate}
    \item Running $2^t$-thresholded BFS from all the sources. Each node will wait until it is done with constructing layered sparse $2^{t+6}$-cover before it can participate in this $2^t$-thresholded BFS.

    \item Constructing sparse $2^{t+7}$-cover. We will simply apply the algorithm from \cref{theorem:sparseasynch}, with the same restriction: a node will be allowed to participate in it only after it is done with constructing layered sparse $2^{t+6}$-cover.
\end{enumerate}

\paragraph{Step 3: Termination.} If all the nodes knew $D$, we could simply run this algorithm for $\lceil \log_2{D} \rceil$ iterations. However, they do not know it. Even though every node will learn the distance to the root in the first $\lceil \log_2{D} \rceil$ iterations, they will keep doing these iterations, blowing up the time and message complexity: since running $t$-th iteration takes $O(2^t \cdot \log^{11}{n})$ time and $O(m \log^8{n})$ messages for constructing the next sparse cover, we cannot afford to simply run $\Theta(\log{n})$ iterations. There are two strategies to solve this problem. In both of them, we are looking for a way to inform every node when all the nodes have already been reached by the BFS.

\paragraph{Approach 1 for step 3.} In \cref{subsec:sparseCovers}, we mentioned that the sparse $d$-covers that we are constructing satisfy the following condition: for each node $v$, there exists a cluster such that all nodes $u$ for which $dist(u, v)\leq d$ are included in this cluster. Thus, by the $\lceil \log_2{D} \rceil$-th iteration, some cluster of some sparse cover that we have constructed will contain all the nodes of the graph.

For each cluster of the sparse $2^t$-cover, add an extra convergecast: checking whether it contains all the nodes of the graph. For each cluster, every node in it checks whether all its neighbors are in the same cluster and propagates this information up. After the root of the cluster of this cover learns whether this cluster contains all the nodes of the graph, it will broadcast this information to all the nodes in the cluster. Every node waits until it receives such verdicts from all its clusters. Note that if no cluster contains all the nodes, every node will receive only ``NO''s; otherwise, every node will receive a ``YES'' from the cluster that contains all the nodes.

If some cluster of the $2^t$-cover contains all the nodes, there is no point in constructing sparse $2^{t_1}$-covers for larger $t_1$: we can simply use the $2^t$-cover instead, so we will not have to waste extra time and messages on constructing those covers. Additionally, after each iteration starting from this one, for every cluster that contains all nodes, we will collect, with a convergecast, whether all the nodes have already been reached by the BFS by the end of that iteration. When the root of the cluster detects this, it will broadcast this to all the nodes in the cluster. Every node waits for all these responses before proceeding to the next iteration. If all nodes have already been reached by the BFS, every node will learn this, and they will simply not go to the next iteration.

\paragraph{Approach 2 for step 3.}  Let us add a slight modification to the $d$-thresholded BFS. Each node at a distance $d$ from the source will check whether any of its neighbors is at a distance larger than $d$ from the source. They will send this information up the execution tree with a convergecast. This way, a source will learn whether the farthest distance from it is larger than $d$. Then, the source will broadcast this information down to all the nodes at its execution tree.

This way, every node will learn whether the largest distance from the source to some other node is at least $d$: if it was reached by BFS, it will learn this from the source. Otherwise, it already knows that it is at a distance larger than $d$ from the source.

So, at the end of the $2^t$-thresholded BFS, every node will learn whether they have to proceed to the next iteration ($2^{t+1}$-thresholded BFS and construct the next sparse cover). We will terminate in $\lceil \log_2{D} \rceil + O(1)$ iterations.

\paragraph{Time complexity.} The isolated time complexity of the initialization is $O(\log^{10}{n})$; the isolated time complexity of the $t$-th iteration is $O(2^t\cdot \log^8{n})$ for BFS plus $O(2^t \cdot \log^{11}{n})$ per cover construction. For the first approach, the convergecasts in the cluster trees add only $O(2^t\log^7{n})$ time messages per iteration; for the second approach, the convergecast in the execution tree adds only $O(2^t)$ time per iteration. In both cases, the total time complexity becomes $O(D \cdot \log^{11}{n})$, as we are doing $O(\lceil \log_2{D} \rceil + O(1)$ iterations.

\paragraph{Message complexity.} The total message complexity of the initialization is $O(m\log^{10}{n})$; the message complexity of each iteration is $O(m\log^5{n})$ for BFS, $O(m\log^8{n})$ for sparse cover construction. plus $O(m)$ for convergecast. For the first approach, the convergecasts in the cluster trees add only $O(m\log^4{n})$ messages per iteration; for the second approach, the convergecast in the execution tree adds only $O(m)$ messages per iteration. In both cases, the total message complexity becomes $O(m\log^{10}{n})$.
\end{proof}

\begin{theorem}
    \label{theorem:MSBFSfast}
    There is an algorithm for a multiple-source BFS that terminates in time $O(D_1 \log^{11}{n})$, where $D_1$ is the largest of the values $dist(v, S)$ over all nodes $v$, and uses $O(m\log^{10}{n})$ messages.
\end{theorem}

\begin{proof}
By applying \textbf{Approach 1} of \cref{theorem:SSBFSfinal}, we would get an algorithm that terminates in $O(D \log^{11}{n})$ time, where $D$ is the diameter of the graph, and uses $O(m\log^{10}{n})$ messages. But we can do better. 

Let us start with some intuition: In the case of the multi-source BFS, the first two stages of the algorithm from \cref{theorem:SSBFSfinal} work the same. The differences start with the termination. Let $D_1$ be the largest of the values $dist(v, S)$ over all nodes $v$, where $S$ is the set of sources. In the synchronous environment, our algorithm would run in time $O(D_1)$. However, the termination \textbf{Approach 1} from \cref{theorem:SSBFSfinal} only guarantees that we finish in time $\tilde{O}(D)$, not $\tilde{O}(D_1)$, while for the \textbf{Approach 2} it is crucial that there is only one source so that all the nodes will be in the same execution tree when they are reached, and we can run a convergecast. Fortunately, there is a nice modification of \textbf{Approach 2}, that allows to run the multi-source in $\tilde{O}(D_1)$ in this case too.

We now present the proof. We will once again do the BFS and build the sparse covers in iterations, but here we would need an extra trick. 

Let us call all nodes \textbf{alive} initially. If a source runs a $d$-thresholded BFS, then it will collect the same information as in the \textbf{Approach 2} from the \cref{theorem:SSBFSfinal}: from the nodes in its execution subtree at depth $d$ it will collect whether they have neighbors that have not been reached by the BFS. A source becomes \textbf{dead}, if the depth of its execution subtree is less than $d$, or if there is no node of depth $d$ in this subtree that has a neighbor not reached by the BFS. If a source becomes dead, it tells this to all the nodes in its execution subtree, and all of them will become dead too. A node that has not been reached by the BFS knows that it remains alive, and a node that is in a subtree of a source that stays alive also stays alive.

Let us say that we did $d$-thresholded BFS. Consider all the alive nodes and sources. The distance from every alive node to the closest alive source has not changed. So, we can run the next iteration only on these alive nodes.

So, the algorithm would work as follows: each node keeps track of whether it is dead or alive for a particular iteration, and only the alive nodes are involved in constructing the corresponding sparse cover (which will be a sparse cover of the subgraph on only the alive nodes). After $\lceil \log_2{D_1} \rceil$ iterations, all nodes will become dead, and the algorithm will terminate.

The time and message complexity bounds are proved in the same way as in \cref{theorem:SSBFSfinal}.
\end{proof}

%-------------------------------------------------------------------------------
\section{Polylogarithmic synchronizer for event-driven algorithms}
\label{sec:generalization}

In this section, we describe the generalization that extends the scheme described in \Cref{sec:BFS} to a synchronizer for a general event-driven synchronous algorithm, thus proving (formal variants of) \Cref{thm:mainInformal}. The formal theorem statements appear later as \Cref{theorem:extensiontimeunknown} and \Cref{theorem:extensiontimeknown}.

\subsection{The general structure of the event-driven algorithms} 
Let us briefly recall the interpretation of event-driven synchronous algorithms, as discussed in \cref{subsec:modelSubtleties}. In event-driven algorithms, nodes cannot refer to the round numbers, i.e., the nodes cannot explicitly access the clock value. A node may send a message only because of receiving some messages (or sending some prior messages and getting their acknowledement). In other words, a node cannot simply wait for several rounds before sending a particular message.

In the synchronous environment, every message is sent at a certain pulse. The messages of pulse $0$ are sent by some \textbf{initiators}; these messages can be perceived as triggered by the outside environment. Every message of pulse $p>0$ has to be triggered by some message of pulse $p-1$. 
It can be triggered in two possible ways:

\begin{itemize}
    \item Node $v$ has received some messages of pulse $p-1$, which triggers sending a message of pulse $p$
    \item Note $v$ has sent some messages of pulse $p-1$, which triggers sending a message of pulse $p$
\end{itemize}

\subsection{The setup}

Let $\alg$ denote the synchronous algorithm that we are trying to run in the asynchronous environment. As before, let $T(\alg)$ denote its runtime, and $M(\alg)$ denote its message complexity.

We will create a virtual node $(v, p)$ for any node $v$ sending one or more messages at pulse $p>0$. Consider any such node $v$, sending a message $m$ with pulse $p>0$. This message is triggered by some message of pulse $p-1$. Choose any of them, let's say, message $m_1$ of pulse $p-1$, from node $u$. Then, let's declare node $(u, p-1)$ as a \textbf{parent} of the node $(v, p)$. Note that it's possible that $u = v$, if $v$ sending a message at pulse $p$ is triggered by $v$ sending a message at pulse $p-1$.

%For every message $m$ of pulse $p>0$, consider the messages of pulse $p-1$ that it was triggered by. Declare an arbitrary one of them as the \textbf{parent} of message $m$.

%We will create a \textbf{virtual} node for every message. More specifically, if a node $u$ sent a message $m$ of pulse $p$ to the node $v$, we will represent this message by a virtual node $(u, v, p)$. 

For any $p$ for which node $u$ is sending a message at pulse $p$, node $u$ will be responsible for performing all the actions of the virtual node $(u, p)$. We describe the exact actions later.

From now on, the setting is very similar to the setting in \cref{subsec:singleBFS}; we define all the same notions, but for the virtual nodes now. Let's define the \textbf{pulse} of a virtual node $(u, p)$. In some sense, our goal is to ensure that $pulse((u, p)) = p$.

Node $v$ is keeping track of the messages that it has received/sent. As soon as it is triggered by some of them (and some other actions of our synchronizer) to send some message $m$ at pulse $p$, it creates a virtual node $(v, p)$. $(v, p)$ picks any of the messages of pulse $p-1$ that it was triggered by, say, $m_1$, sent from $u$, and declares $(u, p-1)$ as its \textbf{parent}. Virtual node $(v, p)$ will then set $pulse((v, p)) = pulse((u, p-1)) + 1$.

So far, everything looks very similar to the \cref{subsec:singleBFS}, but there are several important distinctions. Before even discussing the time and message complexity and the way to actually simulate the work of these virtual messages, we have to understand what the desired correctness guarantees are. In the case of BFS, we had to show that the execution tree will actually be the BFS tree: that the pulse that each node receives will be equal to its distance from the closest source. What do we have to show in this case?

Note that each node acts as a state machine.  It does not know anything except what other nodes tell it. So, its messages depend only on the messages it receives from them. If it receives these messages in some different order, it might decide to send very different messages. For example, for the BFS, for the correct functioning of the algorithm, it is important that the first ``join'' request that a node $v$ at a distance $p>0$ from the sources receives is from a node at a distance $p-1$ from the sources. If this condition is not satisfied, the algorithm will be completely broken. 

Consider a concrete example. Let us say that the node $v$ in the synchronous environment receives messages $m_1, m_2, m_3$ at pulses $p, p+1, p+2$ correspondingly. After receiving $m_2$, it generates a new message $r$. Then, it may be important for $v$ to receive messages $m_1, m_2, m_3$ precisely in this order. Indeed, it might be the case that $v$ would want to generate some other message $r_1$ if it has received only message $m_2$ (without $m_1$), and we have to avoid this. Similarly, it might be the case that $v$ would want to generate some other message $r_2$ if it has received only messages $m_1, m_3$ (without $m_2$), and we have to avoid this. The main takeaway is this: when node $v$ decides to send a message based on having received some messages $m_1, m_2, \ldots, m_k$, it has to be sure that these $m_1, m_2, \ldots, m_k$ correspond to some valid prefix of messages that it would have received in a synchronous environment.

\subsection{A synchronizer, given layered sparse $O(T(\alg))$-cover}
\label{subsec:extensionsynchwithcovers}
\subsubsection{Algorithm}

Let's first consider the case when we are given a layered sparse $(2^6\cdot T(\alg))$-cover. We will try to run almost the same algorithm as in \cref{subsec:multiBFS} on the virtual nodes, with a few small changes. As a reminder, the physical node responsible for the actions of the virtual node $(u, v, p)$ is simply a node $u$. We discuss the adaptations needed to make below.

\paragraph{Notion of distance.} For virtual nodes $(u_1, p_1), (u_2, p_2)$, we define $dist((u_1, p_1), (u_2, p_2))$ as equal to $dist(u_1, v_1)$. Note that it follows that $dist(parent(v), v) \le 1$ for any virtual node $v$.

%\paragraph{Notion of a host.} In the BFS case, every node had a pulse, so we could define $h(v)$ as the ancestor of $v$ with pulse $prev(pulse(v))$. In the current setting, only messages have pulses. The host of message $m$ will then be defined as follows: consider the ancestor message of $v$, which has pulse $prev(pulse(m))$. $h(v)$ will

\paragraph{Sending messages.} In the BFS case, node $v$ was sending ``join'' requests to all its neighbors. In the current setting,  the virtual node $(v, p)$, sending some messages at pulse $p$, will send only the corresponding messages. %Note that even if receiving $m$ triggers $v$ to send some message $m_1$, $m_1$ hasn't existed at the time when $m$ was sent.

\paragraph{Simulation of virtual nodes by the physical nodes.} Physical node $v$ simulates all the actions of the virtual node $(v, p)$ directly: as if it was actually sending those through the same edges. We will point out just one subtlety. If $v$ has a choice, the message from which virtual node to send through a given edge, it chooses the order that would make \cref{corollary:parallel} and \cref{lemma:sequential} still hold. In particular, between messages corresponding to different pulses, if prioritizes smaller ones; for the same pulse, it still does it in a round-robin fashion.

\paragraph{Creation of the virtual nodes.} When a physical node $v$ receives/sends a message of pulse $p-1$, it does the following. First, it waits until it receives $Go\_Ahead(p)$. Then, it looks at all the messages with a pulse at most $p-1$ that it has received/sent. It determines whether, in the synchronous world, it would send any messages at pulse $p$, given the history of the received and sent messages up to this point. Then, if it would send messages at pulse $p$, it creates a virtual node $(v, p)$. Finally, it chooses any virtual node of pulse $p-1$ from which it has received a message as a parent. It then notifies all these virtual nodes of pulse $p-1$ whether they have been chosen. 

\paragraph{Convergecast and broadcast in the execution tree.} Consider the process of collecting the information about the $p$-safety/emptiness of all (virtual) nodes with pulses in $[prev(prev(p)), p]$ for some pulse $p>0$. In the BFS problem, it was easy, as every physical edge appeared only once in the execution tree. Now, a particular physical edge may appear in the execution tree many times, as there might be many messages sent through some particular edge. However, there is still at most one such message per pulse, so this does not lead to any problems, as we see later.

%For every pulse $p_1 \in [prev(prev(p)), p-1]$, we would consider a separate stage, during which every physical node, say, $v$, waits until all the virtual nodes of pulse $p_1$ that it is responsible for, finish collecting their $p$-safety/emptiness information. Only then $v$ transmits this information up the execution tree, for each of these virtual nodes.

%Similarly, during the broadcast of $Go\_Ahead(p)$, node $v$ will wait until it hears from the parents of all the virtual nodes of pulse $p_1$ that it is responsible for. Only then will it propagate the $Go\_Ahead$ down.

%\paragraph{Registration in clusters.} For cluster registration, we cannot afford to simulate the registration of each virtual node. We simply do not have a mechanism to simulate the registration of the same physical node several times as instances of different virtual nodes. So, we will do the following. 

%For every pulse $p>0$, a physical node $v$, which is responsible for some virtual nodes with pulse $prev(prev(p))$, will wait until all of them become $prev(p)$-safe. If one of these virtual nodes is not $p$-empty, it will $p$-register in the cluster. Similarly, it will wait until all of them become $prev(prev(p))$-safe, and only then will $p$-deregister.

\paragraph{The final convergecast for node termination.} In the $d$-thresholded BFS we had to make sure that every node at a distance larger than $d$ from the sources learned it. Here, we do not require any termination of this form.

Everything else remains the same.

\subsubsection{Analysis of correctness}

\cref{lemma:prevclose}, \cref{lemma:closehosts}, \cref{lemma:waiting} translate into the current setting as they are.

\begin{lemma}[Counterpart of \cref{lemma:executionisbfs}]
    \label[lemma]{lemma:executioniscorrect}
    For any physical node $v$, the following claims hold:

    \begin{itemize}
        \item For any messages $m_1, m_2$ with pulses $p_1, p_2$ correspondingly with $p_1<p_2$, $v$ will receive $m_1$ before it receives $m_2$
        \item For any pulse $p>0$, $v$ will receive $Go\_Ahead(p)$ only after it has received all the messages of pulse $p-1$.
    \end{itemize}
    
    %Intuitively, this means that each node can keep track of the states that it would have had in the synchronous environment.
\end{lemma}

\begin{proof}

    \paragraph{First claim.} Let message $m_2$ be sent from (virtual) node $v_2$ of pulse $p_2$, and message $m_1$ be sent from (virtual) node $v_1$ of pulse $p_1$. Suppose that $v$ receives message $m_2$ before $m_1$. Consider the moment when $v$ has already received message $m_2$ but hasn't received $m_1$ yet. 
    
    Let $x$ be the integer from the segment $[p_1+1, p_2]$ divisible by the largest power of $2$ (with largest $\ell(x)$), and let $w$ be the ancestor of $v_2$ with pulse $x$. As the message from $w$ in this execution tree has already been sent,  $w$ has already received $Go\_Ahead(x)$, which is possible only after $h(h(w))$ receives $Go\_Ahead(x)$.

    Now note that $pulse(h(w)) \le p_1$. Indeed, $pulse(h(w)) \le x - 2^{\ell(x)}$, and as number $x - 2^{\ell(x)}$ is divisible by larger power of $2$ than $x$ is, it can't be in $[p+1, p_1]$ by our choice of $x$. So, $pulse(h(w)) \le x- 2^{\ell(x)} \le p$.

    Let $y_1$ be the ancestor of $v_1$ with pulse $prev(prev(x))$, $y_2$ be the ancestor of $v_2$ with pulse $prev(prev(x))$ (so that $y_2 = 
    h(w)$). As $pulse(h(w)) \le p$, $y_1$ is not $prev(x)$-empty. $y_2$ is not $x$-empty by the choice of $x$.

    Next, we show that $dist(y_1, y_2) \le 2^{\ell(p)+5}$ with the argument completely mirroring the argument from \cref{lemma:executionisbfs}. Then, by \cref{lemma:waiting}, $y_2 = h(h(w))$ can receive $Go\_Ahead(x)$ only after $y_1$ has $x$-deregistered, which is not possible until $m_1$ has been sent, contradiction. 

    \paragraph{Second claim.} Suppose the opposite. Suppose that $v$ has already received $Go\_Ahead(p)$ from some virtual node $v_1$ of pulse $p-1$ before it receives a message $m_2$ of pulse $p-1$ from some virtual node $v_2$ of pulse $p - 1$. Consider this moment. $h(h(v_1))$ must have already received $Go\_Ahead(p)$.

    As $dist(v_1, v_2) \le 2$, by part (d) of \cref{lemma:prevclose} we get $dist(h(h(v_1)), h(h(v_2))) \le dist(h(h(v_1)), v_1) + dist(v_1, v_2) + dist(v_2, h(h(v_2))) \le 2 + 2\cdot (9\cdot 2^{\ell(p)}) \le 2^{\ell(p) + 5}$, so we can apply \cref{lemma:waiting}. As $h(h(v_1))$ has already received $Go\_Ahead(p)$, $h(h(v_2))$ has already deregistered from all their common clusters, which happens only after $m_2$ has already been sent.

\end{proof}

\begin{theorem}
    This algorithm correctly synchronizes algorithm $\alg$.
\end{theorem}

\begin{proof}
    Our goal is to show that the set of messages sent by all nodes is exactly the same as it would have been in the synchronous environment. Consider the smallest $p$ such that one of the following happens: some node $v$ that had to send a message $m$ of pulse $p+1$ has not sent it, or some node $v$ sent a message $m$ of pulse $p+1$ that it would not send in the synchronous setting.

    Suppose that the first case happened. In the synchronous setting, $m$ is triggered by some message $m_1$ of pulse $p$. From \cref{lemma:executioniscorrect}, by the time $v$ receives $Go\_Ahead(p+1)$, it has received all the messages of pulse $p$, and, therefore, all the valid messages of pulses up to $p$. Then $v$ would be able to detect that in the synchronous setting, it would have sent $m$, so it would send $m$ now.

    Suppose that the second case happened. Suppose that $m$ has sent some message $m$ of pulse $p+1$ that it would not send in the synchronous setting. This action must have been triggered by receiving $Go\_Ahead(p+1)$. But then, once again, by that time, the set of messages with pulses up to $p$ that $v$ received would coincide with the set it would receive in the synchronous setting, so $v$ would correctly detect that is does not send $m$. 
\end{proof}

\subsubsection{Analysis of the time and message complexity}

The main result of this section is the following theorem:

\begin{theorem}
    \label[theorem]{theorem:extensioncomplexity}
    For an algorithm $\alg$ with worst-case time complexity $T(\alg)$ and message complexity $M(\alg)$, this synchronized version,  given a layered sparse $O(T(\alg))$-cover, finishes in $O(T(\alg) \cdot \log^7{n} \cdot (\log{T(\alg)} + \log{n}))$ time and takes $O(M(\alg) \cdot \log{n}^4 \cdot \log{T(\alg))})$ messages.
\end{theorem}

The proof follows the proof of \cref{theorem:MSBFS} closely. 
We divide all our pulse-related operations into stages, one stage $\stage_p$ for each pulse $p$. Then, we note that the isolated time complexity of $\stage_p$ is still $O(2^{\ell(p)} \log^7{n})$. Due to the \cref{lemma:summation}, the total isolated time complexity of all stages $\stage_i$ becomes $\sum_{p = 1}^{T(\alg)} O(2^{\ell(p)}\cdot \log^7{n}) = O(T(\alg) \cdot \log^7{n} \cdot \log{T(\alg)})$. The initial convergecast for registering all the originators takes $O(T(\alg) \cdot \log^8{n})$ time. So, the entire runtime is $O(T(\alg) \cdot \log^7{n} \cdot (\log{T(\alg)} + \log{n}))$. 

Now the message complexity. The number of messages sent for building the execution tree is $O(M(\alg))$. The number of messages propagated in the execution tree is $O(M(\alg)\log{T(\alg)})$, as, by \cref{lemma:pulsesperedge}, there are messages of only $O(\log{T(\alg)})$ pulses propagated through each edge. Note the difference: in the BFS there were $O(n)$ edges in the execution tree, now $O(M(\alg))$. 

Now let's consider the messages sent through the clusters for registration/deregistration. Once again, for a given level $\ell$, the total number of nodes who will $p$-register for some pulse $p$ with $\ell(p) = l$ is $O(\frac{M(\alg)}{2^{\ell}})$. Such nodes will register in clusters of $2^{\ell+5}$ covers and therefore take $O(2^{\ell}\log^3{n})$ messages to the root and back. Every node is in $O(\log{n})$ clusters, so there will be $O(2^{\ell}\log^3{n} \cdot \log{n} \cdot \frac{M(\alg)}{2^{\ell}}) = O(M(\alg)\log^4{n})$ messages corresponding to registration at a given level, and the total number for all levels is $O(M(\alg) \cdot \log{n}^4 \cdot \log{T(\alg))})$.

\subsection{A synchronizer, without being given layered sparse $O(T(\alg))$-cover}

In the \cref{subsec:extensionsynchwithcovers} we have constructed a synchronizer that works assuming that the layered sparse $O(T(\alg))$-cover is already provided. In this section, we discuss the ways to synchronize from scratch. 

With the algorithm from \cref{subsec:completeBFS}, we know how to build layered sparse $T(\alg)$-covers from scratch. However, it is possible that the exact value of $T(\alg)$, or even its multiplicative approximation, is not known to us. We provide slightly different synchronization approaches in cases when this value is known to us versus when it is not known.

\begin{theorem}
    \label[theorem]{theorem:extensiontimeunknown}
    Consider an algorithm $\alg$ with worst-case time complexity $T(\alg)$ and message complexity $M(\alg)$. Suppose that the nodes \textbf{do not know} the bound on $T(\alg(M))$. There is a synchronizer for $\alg$ such that 

    \begin{itemize}
        \item The entire algorithm terminates in $O(D\log^{11}{n} + T(\alg) \cdot \log^7{n} \cdot (\log{T(\alg)} + \log{n}))$ time
        \item All the outputs of $\alg$ are produced by the time $O(T(\alg)\cdot \log^{11}{n})$. Thus, with the notion of time complexity as the time until all nodes generate their output, the synchronizer has polylogarithmic overhead. \footnote{See \Cref{subsec:modelSubtleties} for discussion on the time complexity definition.}
        \item It takes $O(m\log^{10}{n} + M(\alg)\cdot \log^5{n})$ messages.
    \end{itemize}
\end{theorem}

\begin{proof}

The algorithm will be very similar to the algorithm from the \cref{theorem:SSBFSfinal}. All nodes will start by constructing the initial layered sparse $2^{O(1)}$-cover, then proceed to the iterations. One iteration will consist of simulating the next $2^t$ pulses of $\alg$ and constructing the next sparse $2^{t + O(1)}$-cover.

The question, once again, is how the nodes learn that it's time to stop constructing the sparse covers. We will do this with the \textbf{Approach 1} from the proof of \cref{theorem:SSBFSfinal}. By the $\lceil \log_2{D} \rceil$-th iteration, some constructed cluster would contain all the nodes. It would then notify all the nodes that there is a cluster that contains all the nodes, and they will not construct the next sparse cover: they will use the sparse cover containing this cluster that includes all the nodes instead. 

The isolated time complexity of the initialization remains $O(\log^{10}{n})$ and the isolated time complexity of constructing $2^t$-sparse cover remains $O(2^t \cdot \log^{11}{n})$. Furthermore, after the first $\lceil \log_2{D} \rceil$ iterations, no sparse covers will be constructed. So, the total isolated runtime spent on the sparse cover construction is $O(D\log^{11}{n})$. According to the \cref{theorem:extensioncomplexity}, the actual synchronization of $\alg$ in isolation would take only $O(T(\alg) \cdot \log^7{n} \cdot (\log{T(\alg)} + \log{n}))$ time, so the total runtime is $O(D\log^{11}{n} + T(\alg) \cdot \log^7{n} \cdot (\log{T(\alg)} + \log{n}))$. 

All the messages of $\alg$ will be delivered after the first $\lceil \log_2{T(\alg)} \rceil$ iterations are over; this will take $O(T(\alg)\cdot \log^{11}{n})$.

Moving on to the message complexity, sparse cover construction takes only $O(m\log^{10}{n})$ messages. The registration/deregistration process for the level $\ell$ take $O(\frac{M(\alg)}{2^{\ell}} \cdot min(D, 2^{\ell})\log^3{n} \cdot \log{n})$ time. If $T(\alg) < D$, this sums up to $O(M(\alg)\cdot \log^4{n} \cdot \log{T(\alg)})$, otherwise, this sum telescopes to $O(M(\alg)\log^4{n} \log{D})$. In both cases, we will bound it by $O(M(\alg)\log^5{n})$. So, the message complexity is $O(m\log^{10}{n} + M(\alg)\log^5{n})$.

\end{proof} 

\begin{theorem}
    \label[theorem]{theorem:extensiontimeknown}
    Consider an algorithm $\alg$ with worst-case time complexity $T(\alg)$ and message complexity $M(\alg)$. Suppose that all the nodes \textbf{know} the bound on $T(\alg(M))$. There is a synchronizer for $\alg$ such that 

    \begin{itemize}
        \item All the outputs are produced and the entire algorithm terminates in $O(min(D, T(\alg))\log^{11}{n} + T(\alg) \cdot \log^7{n} \cdot (\log{T(\alg)} + \log{n}))$ time
        \item It takes $O(m\log^{10}{n} + M(\alg)\log^5{n})$ messages.
    \end{itemize}
\end{theorem}

\begin{proof}
    The algorithm remains the same as in \cref{theorem:extensiontimeunknown} with a single modification: all nodes will simply run only the first $\lceil \log_2{T(\alg)} \rceil$ iterations of constructing the sparse cover (but they may still skip going to the next iterations if they detect some cluster that contains all the nodes).
\end{proof}

\section{Applications}
\label{app:applications}

We already described in \Cref{sec:BFS} how to transform the synchronous BFS algorithm to the asynchnous environment so that it runs in time complexity $\tilde{O}(D)$ and message complexity $\tilde{O}(m)$. This already proves \Cref{crl:BFS}. We remark that $\tilde{O}(D)$ is also the time until all nodes terminate. 

We next provide proofs for our other two applications, deterministic asynchronous leader election and minimum spanning tree with near-optimal time and message complexities.

\begin{proof}[Proof of \Cref{crl:LE}] We first describe a relatively simple deterministic leader election algorithm that runs in $\tilde{O}(D)$ rounds and uses  $\tilde{O}(m)$ messages, with the following guarantees: in the end of the algorithm, there is a unique leader elected, each node knows the identifier of the elected leader, and all nodes terminate. 

Initially, all nodes are candidates. The algorithm consists of epochs. Epoch $i$ runs for $\tilde{O}(2^{i})$ rounds and using $\tilde{O}(m)$ messages. During epoch $i$, we first build a sparse $2^{i}$ cover on all nodes, using the synchronous algorithm of \cite{Rozhon2020Polylogarithmic-timeDerandomization}. This algorithm runs in $\tilde{O}(2^{i})$ rounds and generates clusters of depth $\tilde{O}(2^{i})$, with the following properties: (a) each edge is included in $\poly(\log n)$ clusters, with a tree of depth $\tilde{O}(2^{i})$, and (b) for each node $v$, its entire $2^{i}$-hop neighborhood ---the set of all nodes within distance $2^{i}$ of $v$---is included in one cluster. 
Then, in each cluster, using the cluster tree, we perform a convergecast of the minimum candidate identifier, and we then broadcast the value to all nodes of the cluster. This convergecast and broadcast takes $\tilde{O}(2^{i})$ rounds. Any candidate node $v$ whose identifier is not the minimum in one of the clusters that includes $v$ ceases to be a candidate. We then would like to move to the next epoch. But there is a termination detection before that, as we describe next.

We would like to detect termination. Intuitively, we want to terminate as soon as the cluster of a candidate contains all nodes of the graph. To detect this, we do as follows. Each node $v$ sends to its neighbors $u$ all the cluster ids that include $v$. This is doable in $\poly(\log n)$ time and with $\tilde{O}(m)$ messages. Then, we perform a convergecast and broadcast along the cluster tree whether any node $v$ in the cluster $C$ has a neighbor $u$ that is not in cluster $C$. If there is no such node, then the cluster contains all nodes, and the sole remaining candidate --- which is actually the node that had the globally minimum id --- is the leader and each node of the cluster knows the id of the leader.

Notice that in the course of the algorithm, the node with minimum id always remains a candidate. As soon as we reach epoch $i=\lceil{\log_{2} D\rceil}$, there will be one cluster that contains all the nodes in it, and that cluster will detect termination, hence the entire algorithm will terminate. 

The time complexity of the algorithm is $\sum_{i=1}^{\lceil{\log_{2} D\rceil}} \tilde{O}(2^{i}) = \tilde{O}(D)$. The message complexity in each epoch is $\tilde{O}(m)$ and since there are $\lceil{\log_{2} D\rceil}$ epochs until global termination, the overall message complexity is also $\tilde{O}(m)$. 

Finally, this deterministic synchronous leader election algorithm with time complexity $\tilde{O}(D)$ and message complexity $\tilde{O}(m)$ can be directly fed into our synchronizer, hence generating a deterministic asynchronous leader election algorithm with time complexity (and indeeed time to termination) $\tilde{O}(D)$ 
 and message complexity $\tilde{O}(m)$.
\end{proof}

\begin{proof}[Proof of \Cref{crl:MST}] The work of Elkin~\cite{Elkin2020AComplexities} describes a deterministic synchronous MST algorithm that runs in $\tilde{O}(D+\sqrt{n})$ rounds and uses $\tilde{O}(m)$ messages. By invoking our synchronizer on this, we get an asynchronous deterministic algorithm that runs and terminates in $\tilde{O}(D+\sqrt{n})$ time and uses $\tilde{O}(m)$ messages.
\end{proof}

\newpage
\bibliographystyle{alpha}
\bibliography{references, ref3}

\appendix

\section{A review of Awerbuch's $\alpha$, $\beta$, and $\gamma$ synchronizers}
\label{app:alpha-beta-gamma}

As mentioned before, the main challenge in synchronization is the following: consider the time that node $v$ has generated pulse $p$ and should next generate pulse $p$ after it has received every message of pulse $p$ sent to it by each neighbor $u$. The difficulty is that node $v$ does not know which neighbors sent messages to it and which did not, and simply waiting for a certain amount of time cannot resolve the issue; some messages might take an unpredictably long time. 
Generally, we can add an \textit{acknowledgment} to the messages: whenever a node $w$ sends a message to a neighbor $w'$, node $w'$ should send an acknowledgment of that message back to $w'$. This increases message complexity by only a $2$ factor, and waiting for the acknowledgment can increase the time by only a $2$ factor. Let us say node $v$ is \textit{safe} for pulse $p$ when it has received an acknowledgment for each of the messages it sent after its pulse $p$ (and before the next pulses). For a node $v$ to generate its pulse $p+1$, node $v$ needs to ensure that $v$ has received any message of pulse $p$ sent by any neighbor $u$. One way of achieving the latter is to gather the information that ensures that all neighbors of $v$ are safe for pulse $p$.

Awerbuch~\cite{Awerbuch1985ComplexitySynchronization} presented three synchronizer algorithms that follow this outline; he called them $\alpha$, $\beta$, and $\gamma$. These exhibit different tradeoffs between time and message complexity overhead. 

Synchronizer $\alpha$ is the trivial scheme where each node $u$, once $u$ is safe for pulse $p$, sends a message to each neighbor $v$ informing them. Then, node $v$ can generate its pulse $p+1$ once it has received this pulse-$p$ safety information from each of its neighbors. Then, node $v$ can send the messages of round $p+1$ of the synchronous algorithm. This $\alpha$ synchronizer has an ideal $O(1)$ time complexity overhead. However, for each pulse, the scheme communicates safety information over all edges of the graph. Hence it increases the message complexity of algorithm $\mathcal{A}$ to $O(M(\mathcal{A}) + T(\mathcal{A})\cdot m)$. This is a terrible message complexity, as it is asymptotically equal to the highest message complexity possible for the given time complexity.

Synchronizer $\beta$ first uses an initialization phase to build a low-depth tree, e.g., to elect a root and compute a BFS tree rooted at the leader. Then, for each pulse, $\beta$ performs a convergecast and broadcast along the tree so that each node $v$ learns that all other nodes of the graph are safe for pulse $p$. Only after that does the node $v$ generate pulse $p+1$. Hence, this scheme has time complexity overhead $O(D)$, where $D$ denotes the network diameter, and it increases the message complexity to $O(M(\mathcal{A}) + T(\mathcal{A})\cdot n)$. There is also a high time and message complexity for the initialization---computing a leader and building a BFS rooted in it---but we will ignore that here.

Synchronizer $\gamma$ comes close to achieving the best of $\alpha$ and $\beta$. Picking one particularly useful parameterization of the $\gamma$ scheme, it achieves a small time complexity overhead of $O(\log n)$ while increasing the message complexity to only $O(M(\mathcal{A}) + T(\mathcal{A})\cdot n)$. The rough idea is to compute a low-diameter clustering of the vertices, where each cluster has $O(\log n)$-diameter, and keeping only $O(n)$ edges between neighboring clusters such that every two clusters that had adjacent nodes in the graph have at least one edge connecting them. Then, roughly speaking, one applies the $\beta$ scheme inside the clusters and uses the $\alpha$ scheme afterward in between the clusters. This way, each node $v$ will learn when all its neighbors are safe for a pulse $p$ within $O(\log n)$ extra time, and the safety messages of each pulse travel through only $O(n)$ edges, hence implying that the message complexity is increased to $O(M(\mathcal{A}) + T(\mathcal{A})\cdot n)$.

Awerbuch provides only a sequential algorithm for computing such a clustering. However, a follow-up work of Linial and Saks~\cite{Linial1993LowDecompositions} presented a $O(\log^2 n)$-time and $O(m\log^2 n)$ message synchronous algorithm for this clustering. Given the low time complexity, this algorithm itself can be made asynchronous using the $\alpha$ scheme. The combination provides a randomized variant of the $\gamma$ synchronizer that has initialization time complexity $O(\log^2 n)$ and initialization message complexity $O(m\log^2 n)$. 

We comment that this can be made deterministic using recent progress on deterministic synchronous network decomposition, with only $\poly(\log n)$ increase in the parameters. In particular, by plugging in a deterministic network decomposition of Rozhon and Ghaffari~\cite{Rozhon2020Polylogarithmic-timeDerandomization} instead of the algorithm of Linial and Saks~\cite{Linial1993LowDecompositions}, one can obtain a deterministic variant of the $\gamma$ synchronizer, which would turn any synchronous algorithm $A$ with time complexity $T(A)$ and message complexities $M(A)$ into an asynchronous algorithm $A'$ with time complexity $T(A) \cdot \poly(\log n)$ and message complexity $O(M(A) + (m+ T(A) \cdot n)\cdot \polylog n)$. This is including the initialization complexities. Improved variants of the decomposition improve the logarithmic factors~\cite{Ghaffari2021ImprovedDecomposition}.

\paragraph{Optimality of Awerbuch's bounds for generating all pulses at all nodes.} Awerbuch also showed that his bounds are nearly the best possible for the approach of having each node generate each pulse $1$, $2$, $3$, \dots, one per round. In particular, he showed that even for a $2$ round synchronous algorithm, any asynchronous variant that runs in at most $O(k)$ time must have message complexity at least $\Omega(n^{1+1/k})$. 

\paragraph{These synchronizers do not have a small time and message complexity overhead.} Even ignoring the initialization complexities, the synchronizers mentioned above do not necessarily provide good time complexity and message complexity overheads. In particular, consider an algorithm $A$ such that $M(A) \ll T(A)\cdot n$. For such an algorithm, the message complexity overhead $(T(A)\cdot n)/M(A)$ can be quite high.\footnote{We note that many algorithms for important graph problems have this property. Consider, for example, a simple BFS problem, where the synchronous algorithm has time complexity $O(D)$ and message complexity $O(m)$; note that $m$ can be much smaller than $Dn$. See also the message and time-efficient synchronous MST algorithms of \cite{Pandurangan2017ATrees, Elkin2020AComplexities, Haeupler2018Round-andAlgorithms}.} The problem is rooted in trying to generate pulses of all rounds at each node. The synchronous algorithm might have each node send only in few rounds, and it would suffice to generate only those pulses at this node.
\section{Model subtleties}\label{subsec:modelSubtleties}
We presented the synchronous model of distributed message passing in \Cref{subsec:model}. However, there are some subtleties not clarified in this definition; we discuss those here.

\paragraph{Can algorithms reference the time?} One subtlety is whether the algorithm can explicitly refer to the round numbers or not. This can lead to a critical issue for synchronizers. We discuss two natural interpretations of the model. We are not aware of any prior work that draws this distinction. The first interpretation provides a stronger model for algorithms. However, the (randomized) synchronizers of ~\cite{Awerbuch1992AdaptingNetworks, Awerbuch1990NetworkOverhead} do not work for the first interpretation; they would not provide $poly(\log n)$ time and message complexity overheads. But we argue that they work for the second interpretation, and our deterministic synchronizer is also presented for the second interpretation. Fortunately, many graph algorithms of interest adhere to the second interpretation or can be paraphrased to adhere to it without more than a $\poly(\log n)$ increase in time and message complexities.

\paragraph{Clock-based synchronous algorithms.} In the first interpretation of the model, the algorithm can have explicit reference to the round numbers. That is, each node can have explicit access to the clock that counts the rounds, and it can base its action on that clock. For instance, it can do something like: ``in round number \dots, send message \dots"? Per se, this is not forbidden in the model. Indeed many of the standard synchronous algorithms are written in such a way. For instance, the algorithm is made of phases, each phase consisting of a prespecified number of rounds. Each node starts the second phase only after the rounds allotted for the first phase have passed. We call such synchronous algorithms \textit{clock-based}. 

\paragraph{Event-driven synchronous algorithms.} The second interpretation, which is also natural but more restrictive, is that the algorithm cannot explicitly refer to round numbers, i.e., the nodes cannot explicitly access the clock value. The only guarantee is that all messages sent in the network experience the same delay. For instance, this implies that, if there are two paths connecting $v$ to $u$, and node $v$ forwards a message along these two paths, the message will arrive earlier along the shorter path. Algorithms without explicit reference to the clock should base their actions only on new inputs that they receive: ``upon receiving message(s) ..., perform some computation and send message(s) ....'' Note that the response may involve sending several messages, even several to the same neighbor. Hence, the response might take more than a single time unit. However, an algorithm in this interpretation cannot do something like ``send one message, wait for $x$ rounds, and then send another message.'' Fortunately, many graph algorithms of interest can be paraphrased to adhere to this event-driven interpretation, without more than a $\poly(\log n)$ increase in time and message complexities.\footnote{We note that there is also a general transformation. One can transform any clock-based synchronous algorithm to an event-driven one with an $O(1)$ factor time overhead: simply have each node generate a clock for itself by sending a message to one of its neighbors and getting a response. However, this transformation has an additive increase of $nT$ in message complexity, where $T$ denotes the algorithm's round complexity, and this message complexity increase is often too high.}

As mentioned before, the synchronizers of \cite{Awerbuch1992AdaptingNetworks, Awerbuch1990NetworkOverhead} do not work for the clock-based interpretation. They are primarily written for BFS, and then an extension is claimed for general synchronous algorithms. The issue is in the extensions. Consider synchronous algorithms where each node might send several messages (e.g., in response to receiving one message). There are brief sections of explanations for the extension in \cite[Section 6]{Awerbuch1990NetworkOverhead} and \cite[Section 6]{Awerbuch1990NetworkOverhead}. These suggest that we transform the case where one node sends multiple messages to a chain of ``virtual messages" inside the same node, one leading to the next. This is necessary as a part of the ``execution tree" through which all communications get synchronized. In the clock-driven interpretation, if the node sends only two messages but $\Theta(n)$ rounds apart, this chain will include $\Theta(n)$ virtual messages for those two real messages. This can increase the algorithm's message complexity by up to an $\Theta(n)$ factor. Unfortunately, the synchronizer's message complexity depends linearly on the total number of messages, virtual or real, and thus it can have an overhead of $\Omega(n)$ for the message complexity of clock-based synchronous algorithms. Fortunately, the issue seems fixable for the event-driven interpretation of synchronous algorithms, and this is also the interpretation for which we describe our synchronizer. We provide the correctness proof only for our deterministic synchronizer, which subsumes their randomized synchronizer (up to $\log n$ factors).

\paragraph{A subtlety in message delays, in both synchronous and asynchronous models.} Another subtletly is with regard to the message delays. In writing event-driven synchronous algorithms, we need to ensure that no node injects more than one message in each edge per time unit. We can ensure that this way: each node $v$ that sends a message to a neighbor $u$ waits for an acknowledgment from $u$, before $v$ can send a new message to $u$. These acknowledgments increase time and message complexities by at most a $2$ factor. 

In a sense, this subtlety also appears in the definition of the asynchronous model. In defining the asynchronous model, as described in \Cref{subsec:model}, we said that each message injected into the network link takes one (normalized) time unit to arrive at the destination. The value of the time unit is not known to the nodes. But what if the node simultaneously injects $k$ many messages into the link? This has not been explicitly forbidden in the model description. The right interpretation, which takes congestion into account and is closest to the synchronous model, would be as follows: in this batch of simultaneously injected messages, the $i^{th}$ message arrives after $i$ time units. A cleaner resolution is to use the same acknowledgment scheme mentioned above: each node has to wait for an acknowledgment of each message it injected into the link, before injecting the next message. Notice that this way, at each moment in time, each node $v$ can have at most two active messages in its link toward neighbor $u$: one message of the algorithm sent by $v$ to $u$ and at most one acknowledgment for a message initially sent from $u$ to $v$. Node $v$ will not have two acknowledgments active in the link because it does not receive the next message from $u$ before $u$ has received the acknowledgment of the former message it sent to $v$. Thus, again, this acknowledgment mechanism has only a $2$ factor cost in time and message complexities. And it conveniently removes the issue of having to model delays experienced by messages injected into one network link simultaneously or in close succession. 

\paragraph{Time complexity: time to output or time to termination.} Finally, how do we measure time complexity precisely? As mentioned in \Cref{subsec:model}, throughout the paper, we define time complexity as the time from the start of the algorithm to the time that the last node generates its output. We use this for both synchronous algorithms and asynchronous algorithms, and this is indeed the definition for which the $\poly(\log n)$ time complexity overhead claim of \cite{Awerbuch1992AdaptingNetworks} holds (after other fixes). Unfortunately, this point was not made clear in \cite{Awerbuch1992AdaptingNetworks}. A stronger time complexity measure would be to bound the time until all nodes terminate and know that they have terminated. For this stronger definition, both the synchronizers of \cite{Awerbuch1992AdaptingNetworks} and ours will have a larger time complexity overhead. It remains an interesting open problem to develop synchronizers with $\poly(\log n)$ time and message complexity overheads for this stronger definition of time complexity. 

Fortunately, for most applications for concrete graph problems, there is a workaround: (1) The asynchronous algorithm generated by the synchronizers will terminate within $\tilde{O}(D)$ time. Hence, if the synchronous algorithm that is being made asynchronous has time complexity $\Omega(D)$, then the time complexity overhead even for the stronger definition is $\poly(\log n)$. Notice that this is the case for a wide range of applications, especially those for global graph problems, including the breadth first search, leader election, and minimum spanning tree computation. And for synchronous algorithms that run faster---e.g.,$\poly(\log n)$ round synchronous algorithms---the simple $\alpha$ synchronizer described in \Cref{app:alpha-beta-gamma} suffices. (2) If the synchronizer algorithm is given a parameter $\tau$ (i.e., if the value is provided to all nodes), such that the synchronous algorithm is promised to terminate within $O(\tau)$ rounds, then
the asynchronous algorithm will not only generate all the outputs within $\tilde{O}(\tau)$ time but also terminate in all nodes by time $\tilde{O}(\tau)$.
\section{Algorithm for  \Cref{theorem:main_alg_gk}}

\label{subapp:main_alg_gk}

\MAINALGGK*

For the practical needs of this paper, we will briefly recall their algorithm. In the next subsection, we will show that we can run this algorithm in an asynchronous environment with the tools we developed in the previous sections. For proof of the correctness of this synchronous algorithm, we refer to the original paper \cite{Rozhon2020Polylogarithmic-timeDerandomization}.

In the following lemma, we describe the process for constructing the clusters of one color of the network decomposition (e.g., the first color), in a way that it clusters at least half of the vertices.  Since after each application of this lemma only half of the vertices remain, by $\log n$ repetitions, we get a decomposition of all vertices, with $\log n$ colors.

\begin{lemma}\label[lemma]{lemma:main_alg_gk} Let $S\subseteq V$ denote the set of living vertices. There is a deterministic distributed algorithm that, in $O(k \log^9{n})$ communication rounds, finds a subset $S' \subseteq S$ of living vertices, where $|S'|\geq |S|/2$, such that the subgraph $G[S']$ induced by set $S'$ is partitioned into disjoint clusters, so that every two of these clusters are at a distance larger than $k$, and each cluster has weak-diameter $O(k \cdot \log^3 n)$ in graph $G$.

Moreover, for each such cluster $\mathcal{C}$, we have a Steiner tree $T_\mathcal{C}$ with radius $O(k\cdot \log^3 n)$ in $G$ where all nodes of $\mathcal{C}$ are exactly the terminal nodes of $T_\mathcal{C}$. Furthermore, each edge in $G$ is in $O(\log^3{n})$ of these Steiner trees.
\end{lemma}

We obtain \Cref{theorem:main_alg_gk} by $c=\log n$ iterations of applying \Cref{lemma:main_alg_gk}, starting from $S=V$. For each iteration $j\in [1, \log n]$, the set $S'$ are exactly nodes of color $j$ in the network decomposition, and we then continue to the next iteration by setting $S\gets S\setminus S'$. 

In the proof of the lemma, the following observation is useful. Once again, its proof is omitted.

\paragraph{Algorithm for \cref{lemma:main_alg_gk}} The construction has $b=O(\log n)$ phases, corresponding to the number of bits in the identifiers. Initially, we think of all nodes of $S$ as \textbf{living}. During this construction, some living nodes \textbf{die}. We use $S'_i$ to denote the set of living vertices at the beginning of phase $i\in [0, b-1]$. Slightly abusing the notation, we let $S'_b$ denote the set of living vertices at the end of phase $b-1$ and define $S'$ to be the final set of living nodes, i.e., $S' := S'_{b}$. 

Moreover, we label each living node $v$ with a $b$-bit string $\ell(v)$, and we use these labels to define the clusters. At the beginning of the first phase, $\ell(v)$ is simply the unique identifier of node $v$. This label can change over time. For each $b$-bit label $L\in \{0,1\}^b$, we define the corresponding cluster $S'_i(L)\subseteq S'_i$ in phase $i$ to be the set of all living vertices $v\in S'_i$ such that $\ell(v)=L$. 
We will maintain one Steiner tree $T_L$ for each cluster $S'_i(L)$ where all nodes $S'_i(L)$ are the terminal nodes of $T_L$. 
Initially, each cluster consists of only one vertex and this is also the only (terminal) node of its respective Steiner tree. 
%each Steiner tree consists of only one node which is also terminal -- it is the only node in the respective cluster. 

\medskip
\paragraph{Construction invariants.} The construction is such that, for each phase $i\in [0, b-1]$, we maintain the following invariants: 
\begin{enumerate}
    \item[(I')] For each $i$-bit string $Y$, the set $S'_i(Y)\subseteq S'_i$ of all living nodes whose label ends in suffix $Y$ has no other living nodes $S'_i\setminus S'_i(Y)$ in its $k$-hop neighbourhood. 
    
    In other words, the set $S'_i(Y)$ is a union of some connected components of the subgraph $G[S'_i]$ induced by living nodes $S'_i$ and in the $k$-hop neighbourhood in $G$ around $S'_i(Y)$ all nodes are either dead or they do not belong to the set $S$ (they were colored by previous application of the algorithm). 
    \item[(II')] For each label $L$ and the corresponding cluster $S'_i(L)$, the related Steiner tree $T_L$ has radius at most $i \cdot k \cdot R$, where $R=O(\log^2 n)$. 
    \item[(III)] We have $|S'_{i+1}|\geq |S'_i|(1-1/2b)$.
\end{enumerate}

These invariants, together with the observation that each edge is used in $O(\log n)$ Steiner trees, prove \Cref{lemma:main_alg_gk}.

In particular, from the first invariant, we conclude that at the end of $b$ phases, different clusters are at a distance at least $k+1$ from each other. 
From the second invariant, we conclude that each cluster has a Steiner tree with radius $bR=O(k\log^3{n})$. 
Finally, from the third invariant, we conclude that for the final set of living nodes $S'= S'_{b}$, we have $|S'| \geq (1-1/2b)^{b} |S| \geq |S|/2$.

\medskip
\paragraph{Outline of one phase of construction.} We now outline the construction of one phase and describe its goal. Let us think about some fixed phase $i$. We focus on one specific $i$-bit suffix $Y$ and the respective set $S'_i(Y)$. 
%Notice that by the invariant, the labels of all nodes in this component share a common suffix of $i$-bits $Y_{\mathcal{C}}$. 
Let us categorize the nodes in $S'_i(Y)$ into two groups of \textbf{blue} and \textbf{red}, based on whether the $(i+1)^{th}$ least significant bit of their label is $0$ or $1$. Hence, all blue nodes have labels of the form $(*\ldots*0Y)$, and all red nodes have labels of the form $(*\ldots*1Y)$, where $*$ can be an arbitrary bit. During this phase, we make some small number of the red vertices die, and we change the labels of some of the other red vertices to blue labels (and then the node is also colored blue). All blue nodes remain living and keep their label. The eventual goal is that, at the end of the phase, among the living nodes, there is no blue node $b$ and red node $r$ with $dist(b, r) \le k$. This leads to invariant (I) for the next phase. The construction ensures that we kill at most $|{S'}_i(Y)|/{2b}$ red vertices of set $S'_i(Y)$, during this phase. We next describe this construction.

\medskip
\paragraph{Steps of one phase.}
Each phase consists of $R = 10b\log n = O(\log^2 n)$ steps, each of which will be implemented in $O(k\cdot \log^6 n))$ rounds. Hence, the overall round complexity of one phase
%\footnote{We think that one should be able to implement a phase in $O(\log^3 n)$ rounds of the \local model, instead of $O(\log^5 n)$ rounds, by gathering some relevant local topology and then simulating the process locally. However, we leave such optimizations to a later version of this paper.} 
is $O(k\log^8{n})$ and over all the $O(\log n)$ phases, the round complexity of the whole construction of \Cref{lemma:main_alg_gk} is $O(k\cdot \log^9 n)$ as advertised in its statement. 
Each step of the phase works as follows: all blue nodes start a BFS from them up to distance $k$; a node that is reached in the BFS is added to the cluster of the node from which it received the first ``join'' proposal. Next, all red nodes that were added to the cluster will try to join it, to adopt its label. 

For each blue cluster $A$, we have two possibilities: 
\begin{enumerate}
\item[(1)] If the number of adjacent red nodes that requested to join $A$ is less than or equal to $|A|/2b$, then $A$ does not accept any of them and all these requesting red nodes die (because of their request being denied by $A$). In that case, cluster $A$ \textbf{stops} for this whole phase and does not participate in any of the remaining steps of this phase. 
\item[(2)] Otherwise --- i.e., if the number of adjacent red nodes that requested to join $A$ is strictly greater than $|A|/2b$ --- then $A$ accepts all these requests and each of these red nodes change their label to the blue label that is common among all nodes of $A$. In this case, we also grow the Steiner tree of cluster $A$ by one hop to include all these newly joined nodes.
\end{enumerate}

After the breadth first search algorithm finishes, roots of all Steiner trees collect the number of proposing red nodes and each root decides to either accept all proposing red vertices and recolor them to blue, or it makes them die and stops growing. 
The Steiner trees, however, stay the same even if some of its vertices die, the red nodes that died are just labeled as nonterminals. This finishes the description of one step of current phase. 

As each edge is in $O(\log^3{n})$ cluster trees, and diameters of the clusters are $O(k\log^3{n})$, one step can be implemented in $O(k \log^6{n})$ iterations. As there are $O(\log^2{n})$ steps per phase and $O(\log{n})$ phases per color, we get the resulting runtime of $O(k \log^9{n})$.

\end{document}